\documentclass[a4paper,11pt]{amsart}
\usepackage{a4wide}

\usepackage{amssymb,amsmath,latexsym,amsthm}
\usepackage[utf8]{inputenc}

\usepackage[bb=boondox]{mathalfa} 
\usepackage[foot]{amsaddr}
\usepackage{longtable}

\usepackage{mathtools}
\mathtoolsset{showonlyrefs}

\usepackage{wrapfig}
\usepackage{color,mathrsfs}
\usepackage{url}
\usepackage{hyperref}
\usepackage{enumerate}
\usepackage{stmaryrd}

\renewcommand\labelenumi{(\roman{enumi})}
\renewcommand\theenumi\labelenumi

\usepackage[square,sort,comma,numbers]{natbib}

\usepackage[dvipsnames]{xcolor}

\usepackage{wasysym} 
\usepackage{graphicx, rotating}
\usepackage{subfigure}

\newtheorem{theorem}{Theorem}[section]

\newtheorem{definition}[theorem]{Definition}

\newtheorem{Conjecture}[theorem]{Conjecture}

\theoremstyle{definition}
\newtheorem{remark}[theorem]{Remark}

\definecolor{verydarkgreen}{rgb}{0, 0.5, 0}

\usepackage[normalem]{ulem}

\newcommand{\R}{\mathbb R}
\newcommand{\Rd}{\mathbb{R}^d}

\newcommand{\Z}{\mathbb Z}
\newcommand{\N}{\mathbb N}

\DeclareMathSymbol{\tinyvarhexagon}{\mathord}{wasy}{57}

\numberwithin{equation}{section}

\def\XXint#1#2#3{{\setbox0=\hbox{$#1{#2#3}{\int}$}
    \vcenter{\hbox{$#2#3$}}\kern-.5\wd0}}

\allowdisplaybreaks

\title[On the optimality of the rock-salt structure]{\vspace*{-2.5cm}On the optimality of the
  rock-salt structure among lattices with charge distributions}
\author[L.~B\'{e}termin]{Laurent B\'{e}termin}
\address[L.B., M.F.]{University of Vienna, Faculty of Mathematics\\Oskar-Morgenstern-Platz 1, 1090 Vienna, Austria}
\email{laurent.betermin@univie.ac.at}

\author[M.~Faulhuber]{Markus Faulhuber}
\email{markus.faulhuber@univie.ac.at}
\address[M.F.]{RWTH Aachen University, Department of Mathematics, Schinkelstraße 2, 52062 Aachen, Germany}

\author[H.~Knüpfer]{Hans Knüpfer} \email{hans.knuepfer@math.uni-heidelberg.de}
\address[H.K.]{University of Heidelberg, MATCH and IWR, INF 205, 69120 Heidelberg, Germany}

\date{}

\begin{document}

  
\thanks{L.~B\'{e}termin was supported by the Vienna Science and Technology Fund
  (WWTF): MA14-009, by the Austrian Science Fund (FWF) through the project F65 and by VILLUM FONDEN via the QMATH Centre of Excellence
  (grant No. 10059) during his stay in Copenhagen. M.~Faulhuber was partially
  supported by the Vienna Science and Technology Fund (WWTF): VRG12-009 and by
  the Erwin-Schrödinger program of the Austrian Science Fund (FWF):
  J4100-N32. H.~Kn\"upfer was partially supported by the German Research
  Foundation (DFG) by the project \#392124319 and under Germany's
  Excellence Strategy – EXC-2181/1 – 390900948. We thank the anonymous referees for their useful suggestions and comments.}

\begin{abstract}
	The goal of this work is to investigate the optimality of the $d$-dimensional rock-salt structure, i.e., the cubic lattice $V^{1/d}\Z^d$ of volume $V$ with an alternation of charges $\pm 1$ at lattice points, among periodic distribution of charges and lattice structures. We assume that the charges are interacting through two types of radially symmetric interaction potentials, according to their signs. We first restrict our study to the class of orthorhombic lattices. We prove that, for our energy model, the $d$-dimensional rock-salt structure is always a critical point among periodic structures of fixed density. This holds for a large class of potentials. We then investigate the minimization problem among orthorhombic lattices with an alternation of charges for inverse power laws and Gaussian interaction potentials. High density minimality results and low-density non-optimality results are derived for both types of potentials.

	Numerically, we investigate several particular cases in dimensions $2$, $3$ and $8$. The numerics support the conjecture that the rock-salt structure is the global optimum among all lattices and periodic charges, satisfying some natural constraints. For $d=2$, we observe a phase transition of the type ``triangular-rhombic-square-rectangular" for the minimizer's shape as the density decreases.
\end{abstract}

\maketitle
\vspace*{-1.35cm}
\begin{figure}[!htb]
	\includegraphics[width=.23\textwidth]{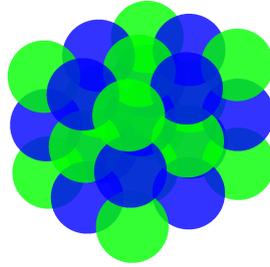}
	\vspace*{-0.4cm}
	\caption{\footnotesize{A rock-salt structure.}}\label{fig_rocksalt}
\end{figure}
\vspace*{-0.8cm}
\tableofcontents

\section{Introduction and statement of main results}

\subsection{Introduction}
The rigorous justification of periodic pattern formations, like crystals or
vortices, in nature and experiments has been extensively studied during the last
decades (see, e.g., \cite{BlancLewin-2015}). In particular, trying to understand
how pure central-forces can generate a low-energy ground-state lattice structure
is a challenging task. The rare existing analytical proofs in dimension
$d\geq 2$, which already exhibit a lot of technical difficulties, treat simplified
models where the geometry of minimizers is explicitely or implicitely
constrained, see for example
\cite{Rad2,Rad3,Crystal,ELi,TheilFlatley,Stef1,Stef2,DelucaFriesecke-2018,Friedrich:2018aa,FrieKreutSquare,BDLPSquare}. However,
the recent breakthrough in dimensions $d\in \{8,24\}$, due to Cohn et
al. \cite{CKMRV2Theta}, has shown the efficiency of Fourier analysis and number
theoretic tools from \cite{CohnElkies,CohnKumar,Viazovska,CKMRV}. With these
tools the authors of \cite{CKMRV2Theta} were able to show the universal
optimality of the $\mathsf{E}_8$ and the Leech lattices among all configurations
of charges with fixed density and absolutely summable interaction potentials $F$
of the form $F(r)=f(r^2)$, with $f$ being a completely monotone function (see
the space $\mathcal{S}$ in Def.~\ref{def-potentials}).

In the present work, we focus on the minimization among, both, periodic charge
distributions $\varphi$ (or any kind of weight associated to the types of
particles) and (simple) lattice structures $L$. All possible charge
distributions $\varphi:L\to \R$ are assumed to be periodic with respect to a
periodicity cube $K_N(L)$ of ``size" $N$, containing $N^d$ points and satisfying
some specific constraints like positivity at the origin and fixed $L^2$ norm on
$K_N(L)$ (see Def.~\ref{def-lattice} for details). We consider the space of $d$-dimensional lattices $\mathcal{L}$, and subspaces $\mathcal{L}(V)\subset \mathcal{L}$ of lattices with a fixed unit cell of volume
$V>0$. The space of all admissible
charge distributions on a given lattice $L$ is called $\Lambda_N(L)$. Our
goal is to show, both, analytically and numerically, that the (properly scaled)
cubic lattice $\Z^d$ with an alternating distribution $\varphi_\pm$ of charges
$\pm 1$ (see Def.~\ref{def-chargedlattice}), is the natural candidate as the
ground state of systems interacting through a large class of radially symmetric
potentials. Throughout this work, we will refer to these structures as the
rock-salt structure. The term is borrowed from chemistry and inspired by
Fig.~\ref{fig_rocksalt}, which illustrates the structure of
a Sodium Chloride crystal (NaCl), also known as rock-salt. For a 2-dimensional
illustration of our model see Fig.~\ref{fig_opt_charge} (a).

More precisely, for two given charges $\varphi(x),\varphi(y)$ at the points
  $x,y$, we consider the pairwise energy of $(x,y)\in \Rd \times \Rd$, given by
\begin{align}\label{intro:potential}
	f_1(|x-y|^2)+\varphi(x)\varphi(y)f_2(|x-y|^2).
\end{align}
Here  $r \mapsto f_1(r^2)$ with
$r = |x-y|$ represents the repulsion between the two
particles and
$r \mapsto  f_2(r^2)$ represents the pure charge interaction. This
interaction is, according to the fact that the term ``charge" has to be
understood in a broad sense, not necessarily Coulombian. We then ask for the
minimizer of the energy per point of the system, defined by
\begin{align}\label{eq-mainenergyintro}
	E_{f_1,f_2}[L,\varphi] \ := \ \sum_{p\in L\backslash \{0\}} f_1(|p|^2)+\frac{1}{N^d} \sum_{p\in L\backslash \{0\}} \, \sum_{x\in K_N(L)}\varphi(x)\varphi(x+p)f_2(|p|^2),
\end{align}
among pairs $(L,\varphi)$ of lattices and admissible charge distributions. In
this paper, only the case where $r \mapsto f_1(r^2)$ and
$r \mapsto f_2(r^2)$ are absolutely summable will be considered, but it is
important to notice that all the results also hold if $f_2$ does not have this
property assuming that the total net charge on the
periodicity cube is zero (see Rmk.~\ref{rmk-nonintegrablecase}). An
important example of a non--integrable potential is the Coulomb
potential, where 
%\begin{align}
  $f_2\left(r^2\right) = \frac{1}{r^{d-2}}$ for $d\geq 3$,
%\end{align}
which is not presented in this paper but for which we have checked that our numerics exhibit the same result as in the absolute summable case.

By $\mathcal{F}$ we denote the set of all functions $f:\R_+\to \R$ which are the Laplace transform of a Borel measure $\mu_f$, such that $f(r)=O(r^{-s})$ for some $s>d/2$ as $r\to +\infty$, and by $\mathcal{S}\subset \mathcal{F}$ we denote the subset of those functions with $\mu_f$ being nonnegative (see Def.~\ref{def-potentials} for details).

%\medskip
%

A first study in the analysis of charged lattices,
concerning only the second term of \eqref{eq-mainenergyintro}, has been
given by two of the authors in
\cite{BeterminKnuepfer-preprint}, proving a conjecture stated by Born
\cite{Born-1921} in 1921: When $f_1\in \mathcal{F}$ and $f_2\in \mathcal{S}$,
for all orthorhombic lattices $L$ (i.e., their unit cells are hyper-cuboids, see
Def.~\ref{def-lattice}), the unique minimizer of
$\varphi\mapsto E_{f_1,f_2}[L,\varphi]$ is the alternate distribution of charges
$\varphi_\pm\in \{-1,1\}$ (see Def.~\ref{def-chargedlattice}). This result
holds for a large class of potentials $f_2$ that are not necessarily
integrable at infinity, which for instance is the case for Coulomb potentials in
dimensions $d\geq 3$. We will write $\mathcal{Q}$
(resp.~$\mathcal{Q}(V)\subset \mathcal{Q}$) for the space of orthorhombic
lattices (resp.~with a unit cell of volume $V>0$). In the triangular lattice
case, where $\mathsf{A}_2$ is the triangular lattice of unit density defined in
Def.~\ref{def-lattice}, a surprising honeycomb-like distribution of charges
$\varphi_{\tinyvarhexagon}$, with charges $+\sqrt{2}$ surrounded by 6 charges
$-\sqrt{2}/2$, has been found as the minimizer (see
Def.~\ref{def-chargedlattice} and Fig.~\ref{fig_opt_charge} (b)). For certain
compactly supported potentials $f_2$, the triangular lattice is actually a local
maximizer of the energy $L\mapsto E_{0,f_2}[L,\varphi]$ for this specific charge
distribution. This is a consequence of the results in \cite{FauSte19}, where the
conditions on the potential are just strong enough to overcome technical
difficulties in the proof. It is plausible to assume that the result would hold
for completely monotonic potentials. Another conjecture here, also mentioned in \cite{FauSte19}, is that the
triangular lattice $\mathsf{A}_2$ with its optimal charge distribution
$\varphi_{\tinyvarhexagon}$ is actually a global maximizer of $E_{0,f_2}$, where
$f_2\in \mathcal{S}$, among all optimally charged lattices
$(L,\varphi)\in \mathcal{L}(1)\times \Lambda_N(\mathcal{L}(1))$. As shown in
\cite{BeterminKnuepfer-preprint}, this general problem of minimizing energies of
type $\varphi\mapsto E_{f_1,f_2}[L,\varphi]$ in $\Lambda_N(L)$ for fixed
$L\in \mathcal{L}$ and $f_1\in \mathcal{F},f_2\in \mathcal{S}$, is equivalent to
finding the coldest point of the heat kernel on a flat torus. This task is
extremely challenging as the location of the coldest point in general depends
not only on the geometry, but also on the volume of the lattice. For research in
this direction we refer to
\cite{Baernstein-1997,BeterminPetrache,Faulhuber_Rama_2019}.
\begin{figure}[!htb]
  \subfigure[Rock-salt structure. The optimal configuration of charges for the square lattice $\Z^2$
 is the alternate distribution $\varphi_\pm$.]  { \includegraphics[width=.35\textwidth]{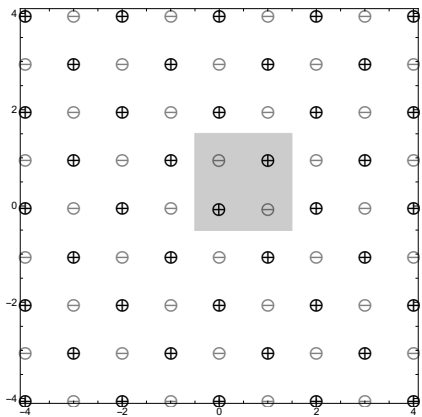}
  } \hfill \subfigure[Honeycomb-like configuration. The optimal configuration of charges for the triangular lattice $\mathsf{A}_2$ is the distribution of charges $\varphi_\tinyvarhexagon$.]  { \includegraphics[width=.35\textwidth]{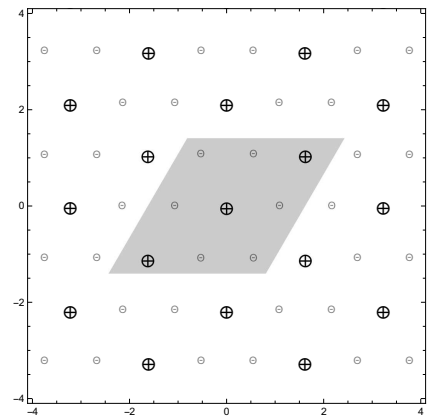} }
  \caption{\footnotesize{Optimally charged lattice structures in dimension 2.
      The total net charge of the marked periodicity cell is $0$ for both lattices.}}\label{fig_opt_charge}
\end{figure}

We first consider minimization problems with only
one type of charges (``positive''). Except for the results in \cite{CKMRV2Theta},
the search for lattices minimizing the energy for radially symmetric potentials
has been restricted to certain small classes of potentials and lattices. For
$d=2$, Montgomery \cite{Mont} (see also \cite{NonnenVoros,AftBN}) proved that
the triangular lattice uniquely minimizes the considered Gaussian energy, also
called lattice theta function (see Def.~\ref{def-epstein}), in $\mathcal{L}(V)$
for all $V>0$ (see Section \ref{sec-setting}). An important consequence is that
the same optimality is true for any completely monotone potential. Therefore,
the same result holds for example for the Epstein zeta-function (see again
Def.~\ref{def-epstein}). For $d \geq 2$, $d \not\in \{8,24\}$, only asymptotic
and local minimality results among lattices of fixed density exist (see, e.g.,
\cite{Ennola,DeloneRysh,SarStromb,Coulangeon:2010uq,Gruber,BeterminPetrache,Beterminlocal3d,CoulSchurm2018}). However,
the behavior of the energy of orthorhombic lattices is well-understood in any
dimension $d \geq 1$ and has also been proved by Montgomery \cite{Mont}: for any
fixed volume $V>0$, the cubic lattice $V^{\frac{1}{d}} \Z^d$ uniquely minimizes
the lattice theta function in the space of orthorhombic lattices
$\mathcal{Q}(V)$. Again, this includes the same result for completely monotone
interactions and, in particular, for the Epstein zeta function as already shown
by Lim and Teo \cite{LimTeo} (see Rmk.~\ref{subsec:optiZdCM} for details).

The same kind of problem has been studied for $d \geq 1$ for orthorhombic
lattices with alternating charges $\pm 1$
\cite{Faulhuber:2016aa,BeterminPetrache} and the strict maximality of $\Z^d$
holds again for any completely monotone function (see again
Rmk.~\ref{subsec:optiZdCM} for details). As previously pointed out, among
orthorhombic lattices, the minimization of the energy among charges $\varphi$
yields an alternate distribution $\varphi_\pm\in \{-1,1\}$, which is universal
for this class (see \cite[Thm. 2.4]{BeterminKnuepfer-preprint}).

 For the general problem \eqref{eq-mainenergyintro} we have to deal with two
competing interactions $f_1$, $f_2$. Then
particles of the same kind interact through the repulsive potential $f_1+f_2$
whereas particles of different kinds interact through the attractive-repulsive
potential $f_1-f_2$ (see Fig.~\ref{fig_intercation_function} for two
examples). Since $f_1$ and $f_2$ do not scale in the same way with respect to
the volume $V$ of the unit cell of the lattices (i.e., the inverse density), the
minimizer of this mixed energy must depend on $V$. In the class of
  orthorhombic lattices $\mathcal{Q}(V)$, the first term of the energy is
  minimized by $\Z^d$, whereas the second term of the energy is maximized by
  $\Z^d$.

\begin{figure}[htb]
	\subfigure[The inverse power law case for $f_1(r^2) = r^{-8}$ and $f_2(r^2) = r^{-6}$ and $r \geq 0.8$.]{
		\includegraphics[width=.45\textwidth]{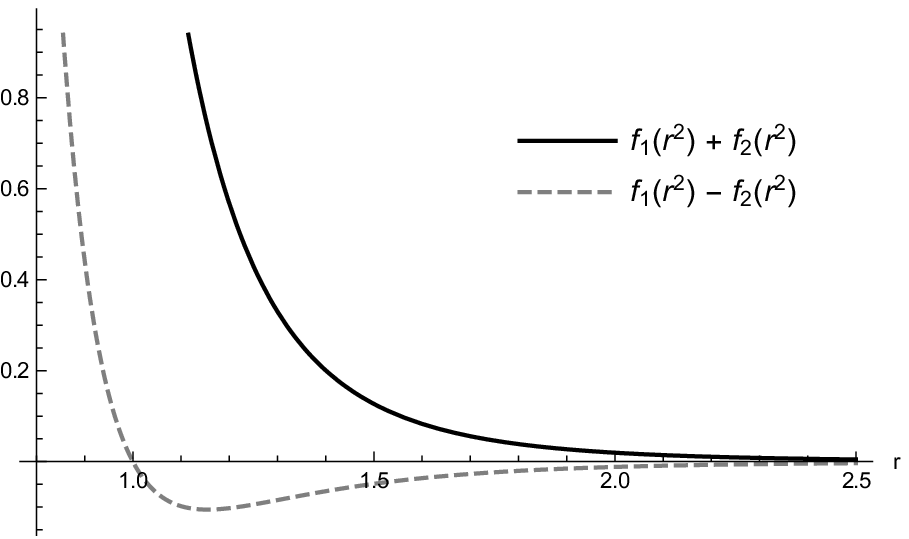}
	}
	\hfill
	\subfigure[The Gaussian case for $f_1(r^2) = e^{-2 \pi r^2}$ and $f_2(r^2) = e^{-\pi r^2}$ and $r \geq 0$.]{
		\includegraphics[width=.45\textwidth]{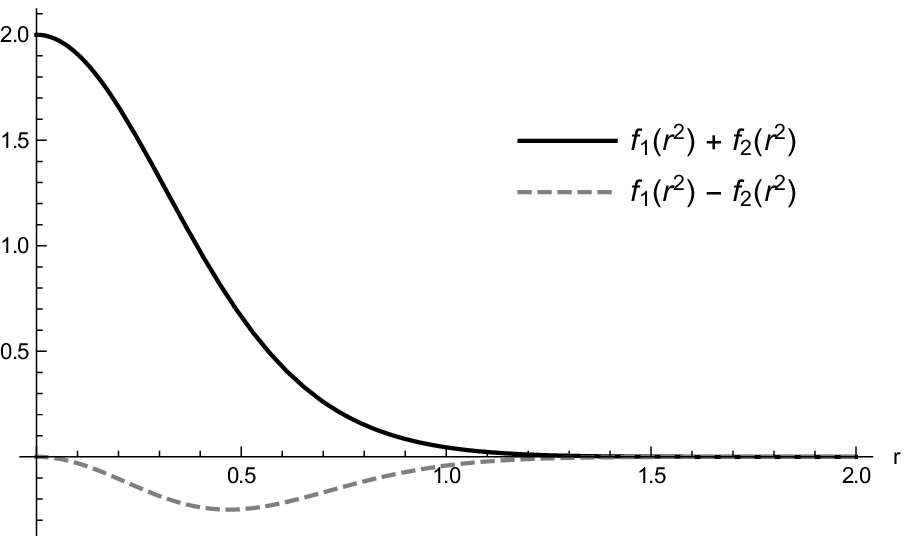}
	}
	\caption{\footnotesize{Interaction potentials for particles of same and different kinds.}}
	\label{fig_intercation_function}
\end{figure}

\subsection*{Structure of the paper} We next present our main results, noting
that the precise definitions and notations are given in Section \ref{sec-setting}. In Section \ref{subsec:design},
we give some preliminary results concerning the criticality of the cubic lattice
and its optimality for completely monotone potentials.  The inverse power law
case and the Gaussian case are treated in Section \ref{sec-invpowerlaws} and in
Section \ref{sec-Gaussian}, respectively. We will also discuss the
(non-)optimality of the $d$-dimensional cubic lattice. In the final Section
\ref{sec-numerics}, we carry out some numerical investigations which support the
Conjecture \ref{con-conj1} on the global optimality of the rock-salt structure
for certain interaction potentials.

\section{Statement of results}\label{sec-results}

Our presentation of the main results uses some basic notions and notations
  for lattices and interaction potentials. These have already been quickly
  characterized in the introduction. For a more precise and thorough description
  of the setting, we refer the reader to the Definitions in the following
  Section \ref{sec-setting}, where the notions of lattices and potentials are introduced rigorously. Furthermore, we give below a list of frequently used notations in the paper, together with the locations where those notations are introduced. We consider this to be helpful for the reader.
\begin{longtable}{rll}
$\mathcal{L}$ & - & set of all lattices (see Def.~\ref{def-lattice})\\
$\mathcal{L}(V)$ & - & set of all lattices with volume $V$ (see Def.~\ref{def-lattice})\\
$\mathcal{Q}$ & - & set of all orthorhombic lattices (see Def.~\ref{def-lattice})\\
$\mathcal{Q}(V)$ & - & set of all orthorhombic lattices of volume $V$ (see Def.~\ref{def-lattice})\\
$\mathsf{A}_2$ & - & triangular lattice (see Def.~\ref{def-lattice})\\
$\Lambda_N(L)$ & - & set of $N$-periodic charge distributions on a lattice $L$ (see Def.~\ref{def-chargedlattice})\\
$K_N(L)$ & - & $N\times N$ periodicity cube on a lattice $L$ (see Def.~\ref{def-chargedlattice})\\
$\varphi^\pm$ & - & alternate charge distributions $\pm 1$ (see Def.~\ref{def-chargedlattice})\\
$\varphi_{\tinyvarhexagon}$ & - & honeycomb-like charge distribution for 2-dimensional lattices (see Def.~\ref{def-chargedlattice})\\
$\mathcal{F}$ & - & space of potentials that are Laplace transforms of a measure (see Def.~\ref{def-potentials})\\
$\mathcal{S}$ & - & space of potentials belonging to $\mathcal{F}$ with a positive measure (see Def.~\ref{def-potentials})\\
$\mathcal{E}_{f_1,f_2}$ & - & general charged lattice energy (see Def.~\ref{def-energy})\\
$\mathcal{E}_{f_1,f_2}^\pm[L]$ & - & energy for an alternation of charges (see Def.~\ref{def-energy})\\
$\zeta_L$ & - & Epstein zeta function (see Def.~\ref{def-epstein})\\
$\theta_L$ & - & lattice theta function (see Def.~\ref{def-epstein})
\end{longtable}
According to the existing results, a first natural step is to consider the
problem of minimizing the energy \eqref{eq-mainenergyintro} with
$f_1,f_2\in \mathcal{F}$ and restricting the structures to be orthorhombic. Since the interactions in a physical multi-component system are rarely only given by potentials -- but also
by quantum effects, entropy, kinetic energy, etc. -- the orthorhombic shape of a
ground state can be first assumed as the consequence of additional effects, like
orthogonal orbital shapes, which do not appear in our model (see for instance
\cite[Sect.~5.1.4]{NaumannBook} for a discussion about metal structures). As
recalled before, a direct application of \cite{BeterminKnuepfer-preprint} is
that the alternate distribution of charges $\varphi_\pm$ minimizes
$\varphi\mapsto E_{f_1,f_2}[L,\varphi]$ in the orthorhombic case and when
$f_1\in \mathcal{F}, f_2\in \mathcal{S}$. It is therefore sufficient to
study the following energy model:
\begin{align}\label{eq-defEf1f2pm}
  E_{f_1,f_2}^\pm[L] \ := \ E_{f_1,f_2}[L,\varphi_\pm] \ %
  = \ \sum_{p\in L\backslash  \{0\}} f_1(|p|^2)+\sum_{p\in L\backslash \{0\}} \varphi_\pm(p) f_2(|p|^2).
      \end{align}
      Using spherical design tools from number theory (in particular from
      \cite{Coulangeon:2010uq}), we prove that $\Z^d$ is always a critical point
      of the energy \eqref{eq-defEf1f2pm}, in the space of general lattices
        with alternate charge distribution:
\begin{theorem}[Criticality of the cubic lattice for general potentials]\label{thm-maincriticZd}
  Let $d\geq 1$. For all $f_1,f_2\in \mathcal{F}$, and all $V>0$, $V^{\frac{1}{d}}\Z^d$ is a critical point in $\mathcal{L}(V)$ of $E_{f_1,f_2}^\pm$ defined by \eqref{eq-defEf1f2pm}.
\end{theorem}
While the above theorem allows for general lattices, we note that in
particular, we also get that the rock-salt structure is a critical point
in the smaller class of orthorhombic lattices.

A typical model for ionic interactions is to consider a power law repulsion at short distance between particles, also called `Born approximation' (corresponding to the Pauli repulsion principle, see, e.g., \cite[Sect.~3.2.2]{NaumannBook}), together with a Coulomb interaction between charged particles. Staying in the space of absolutely integrable potentials, we hence consider pairs of inverse power laws
\begin{align}\label{eq-deff1f2invpower}
	\left(f_1(r), \ f_2(r)\right) \ = \ \left(r^{-p}, \ r^{-q} \right), \qquad  p>q>d/2 .
\end{align}
The energy $E_{f_1,f_2}^\pm$ can then be written in terms of the Epstein zeta function and the alternate Epstein zeta function (see Def.~\ref{def-epstein}). The next three theorems are concerned with the power law case when the potentials are of the form \eqref{eq-deff1f2invpower}. We first show the following result stating the optimality of the rock-salt structure at high density among orthorhombic lattices, as well as its non-optimality among all lattices at low density.
\begin{theorem}[Cubic lattice at high and low density]\label{thm-mainInvpowerlaws}
  Let $d\geq 1$ and let $f_1,f_2\in \mathcal{S}$ be given by
  \eqref{eq-deff1f2invpower}. Then there exist $V_0$ and $V_1$ (both depending
  only on $p$, $q$, $d$) such that the following holds: For all $0<V<V_0$,
  $(V^{\frac{1}{d}}\Z^d,\varphi_\pm)$ is the unique minimizer of $E_{f_1,f_2}$ in $\mathcal{Q}(V)$,
    i.e.,
	\begin{align}
          E_{f_1,f_2}[V^{\frac{1}{d}}L,\varphi] \  \geq \ E_{f_1,f_2}^\pm[V^{\frac{1}{d}}\Z^d], \qquad %
          \forall L \in \mathcal{Q}(1), \, N \in \N, \, \varphi \in \Lambda_N(L) \ .
	\end{align}
	Equality holds if and only if $N\in 2\N$, $\varphi=\varphi_\pm$ and $L=\Z^d$. Furthermore, for all $V>V_1$, $(V^{\frac{1}{d}}\Z^d,\varphi_\pm)$ is not a minimizer of $E_{f_1,f_2}$.
\end{theorem}

This is, as far as we know, the first rigorous result of the optimality (and
  non--optimality) of the rock-salt structure among orthorhombic structures in
any dimension and for arbitrary charge distributions. A key point to prove this
result is that $f_1$ dominates $f_2$ at the origin while $f_2$ dominates
  $f_1$ for large distances. We note that, if we do not restrict the
minimization to the class $\mathcal{Q}(V)$, the cubic lattice is not expected to
be optimal for the minimization among lattices of small fixed volume with an alternate distribution of charges. For instance, in dimension $d=2$, the first
term of $E_{f_1,f_2}^\pm$ is  $\zeta_L(2p)$, which is dominant at high density and minimized by the triangular lattice. Therefore, it will also minimize our energy $E_{f_1,f_2}^\pm$ for alternate charges for certain small values of
$V$. However, it is reasonable to believe that there exists $V_0$ such that
$V_0^{1/d}\Z^2$ is globally minimal. The same remark can be stated in dimensions
$8$ and $24$ as a consequence of the universal optimality of $\mathsf{E}_8$ and
the Leech lattice proved in \cite{CKMRV2Theta}. We also expect the same kind of
result as Thm.~\ref{thm-mainInvpowerlaws} to hold if $f_1$ is replaced by a
Lennard-Jones type potential $f_1(r)=r^{-s}-r^{-t}$, $s>t>q>d/2$ by using the
same arguments based on properties of lattice theta functions presented in \cite{BetTheta15}.

%\medskip

Furthermore, we also study the local optimality of the rock-salt structure in
dimension $d=2$ among orthorhombic (i.e., rectangular) lattices having an
alternating distribution of charges. Any rectangular lattice
$L\in \mathcal{Q}(1)$ can be parametrized by only one real number $y > 0$ via
the form $L=\Z\left(y,0 \right) \oplus \Z\left( 0,y^{-1} \right)$. It is then
easy to find out for which volume the alternate square lattice
$(\Z^2,\varphi_\pm)$ is a local minimum of our energy. Numerical investigations
suggest that the local minimality of the square lattice implies its global
minimality, that is why the following result is useful.
\begin{theorem}[Local optimality in $\mathcal Q(V)$ for inverse power laws]\label{thm-local2dInvPowLaw}
	Let $d=2$ and $f_1,f_2\in \mathcal{S}$ be given by \eqref{eq-deff1f2invpower} and $E_{f_1,f_2}^\pm$ by \eqref{eq-defEf1f2pm}. Then there exists a precise value $V_{p,q}$ (see \eqref{def:Vpq2d} for the formula) such that the following holds:
	\begin{enumerate}
		\item If $V<V_{p,q}$, then $V^{\frac{1}{2}} \Z^2$ is a strict local minimum of $E_{f_1,f_2}^\pm$ in $\mathcal{Q}(V)$.
		\item If $V>V_{p,q}$, then $V^{\frac{1}{2}} \Z^2$ is a strict local maximum of $E_{f_1,f_2}^\pm$ in $\mathcal{Q}(V)$.
	\end{enumerate}
      \end{theorem}
      In particular, in view of Theorem \ref{thm-local2dInvPowLaw} for any
        $V>V_{p,q}$, the square lattice $V^{\frac{1}{2}} \Z^2$ is not a
        minimizer of $E_{f_1,f_2}^\pm$ in $\mathcal{Q}(V)$. Notice that we have chosen to omit the $V=V_{p,q}$ case since it appears too technical to obtain a general result with respect to $p$ and $q$.  Using the
      homogeneity of the Epstein zeta function and the alternating Epstein zeta
      function, given by \eqref{eq_alternate_EZ}, we obtain the following result
      which gives the optimal density for any given lattice.

\begin{theorem}[Minimal energy for a given lattice shape]\label{thm-minenergyzetahomog}
  Let $d\geq 1$, and $p>q>d/2$ and let $f_1,f_2$ be given by
  \eqref{eq-deff1f2invpower}. For $L\in \mathcal{L}(1)$, if $\zeta_L^\pm(2q)<0$, then the function
  $V\mapsto E_{f_1,f_2}^\pm[V^{\frac{1}{d}}L]$ is decreasing on $(0,V_L)$ and
  increasing on $(V_L,+\infty)$ for some $V_L>0$ (see \eqref{def:VL} for a
  formula) and its minimum is
	\begin{align}\label{eq-minenergyL}
		\mathcal{E}^\pm_L \ := \ \min_{V>0} E_{f_1,f_2}^\pm[V^{\frac{1}{d}}L] \ = \ E_{f_1,f_2}^\pm[V_L^{\frac{1}{d}}L] \ = \ \frac{(q-p)\left( -q\zeta_L^{\pm}(2q) \right)^{\frac{p}{p-q}}}{qp(p\zeta_L(2p))^{\frac{q}{p-q}}} < 0.
	\end{align}
	If $\zeta_L^\pm(2q)\geq 0$, then $V\mapsto E_{f_1,f_2}^\pm[V^{\frac{1}{d}}L]$ is strictly decreasing on $(0,+\infty)$ and does not have any minimizer.
\end{theorem}
Therefore, minimizing $L \mapsto \mathcal{E}^\pm_L$ (given by
\eqref{eq-minenergyL}) in $\mathcal{L}(1)$ is equivalent to minimizing
$E_{f_1,f_2}^\pm$ in $\mathcal{L}(V)$, which simplifies the numerical search for
a \textit{global} minimizer.

\begin{remark}[Negativity of alternate lattice sums]
  It is unclear whether $\zeta_L^\pm(s)$ --- and more generally $E_{0,f_2}^\pm$,
  $f_2\in \mathcal{S}$ --- is negative for all dimensions $d$, all lattices
  $L\subset \R^d$ and all $s>d$. We expect this property to hold in low
  dimensions and we did not find any counterexample while checking the FCC, BCC,
  $\mathsf{E}_8$ and Leech lattices. Presently, we only know that this property
  holds for any $f_2\in \mathcal{S}$ and any orthorhombic lattice
  $L\in \mathcal{Q}$. This follows by applying the  integral representation \ref{eq-Ehpmintegral} and the fact that $\theta_4(t)<1$ for all $t>0$.

\flushright{$\diamond$}
\end{remark}

Finally, we also study the special case of Gaussian potentials of the form,
\begin{align}	\label{eq-deff1f2Gaussians}
  \left(f_1(r),f_2(r)\right) \ %
  = \ (e^{-\pi \beta r}, e^{-\pi \alpha r}), \qquad \beta>\alpha > 0,
\end{align}
which is related to many physical systems like Bose-Einstein condensates
\cite{Mueller:2002aa,AftBN} or 3-block copolymers \cite{LuoChenWei} (see
Rmk.~\ref{rmk-BEC}). In this case, for any lattice and any pair of
  functions the energy $E_{f_1,f_2}^\pm$ can be written in terms of the lattice theta and alternate lattice theta functions (see \eqref{eq:theta}).

This model generalizes the problems investigated in
\cite{Mont,BeterminPetrache,Faulhuber:2016aa}, where products involving
$\theta_3$ and $\theta_4$ were studied separately (see also
Rmk.~\ref{subsec:optiZdCM}). This case is also of interest as Gaussian
functions are the building blocks of potentials obtained by the Laplace
Transform of measures (see e.g., \cite{OptinonCM}). For the two-dimensional rock-salt structure, we find its non-optimality at low density, its minimality for small $\alpha$ and for
  fixed $\beta$ and $V > 0$, as well as its optimality at fixed density
when $f_2$ is replaced by $\varepsilon f_2$, $\varepsilon$ small enough.
\begin{theorem}[Optimality of the cubic lattice for Gaussian interactions]\label{thm-mainGaussians}
  Let $d=2$ and let $f_1,f_2$ be defined by \eqref{eq-deff1f2Gaussians}. For
    $V > 0$ and $\beta > \alpha > 0$ we have the following results:
	\begin{enumerate}
		\item There exists $V_1 = V_1(\alpha,\beta)$ such that for all $V>V_1$, $V^{\frac{1}{2}}\Z^2$ is not a minimizer of $E_{f_1,f_2}^\pm$ defined by \eqref{eq-defEf1f2pm} in $\mathcal{Q}(V)$, but a local maximizer.
		\item There exists
                  $\alpha_0 = \alpha_0(\beta,V)$ such that if
                    $\alpha < \alpha_0$ then $V^{\frac{1}{2}}\Z^2$ is the
                  unique minimizer of $E_{f_1,f_2}^\pm$ in $\mathcal{Q}(V)$.
		\item There exists $\varepsilon_0 = \varepsilon_0(\alpha,\beta, V) > 0$	
		such that, if $\varepsilon\in [0,\varepsilon_0)$,
                  $V^{\frac{1}{2}}\Z^2$ is the unique minimizer of
                  $E_{f_1,\varepsilon f_2}^{\pm}$ in $\mathcal{Q}(V)$.
	\end{enumerate}
\end{theorem} 
We believe that these results hold for any dimension, but for simplicity we
prefer to present the proof only in dimension $d=2$. Furthermore, we expect the
variation of $V\mapsto E_{f_1,f_2}^\pm[V^{\frac{1}{d}}L]$ to be similar for the
Gaussian case as presented in Thm.~\ref{thm-minenergyzetahomog} for the inverse
power law case. We have numerically checked this property for some values of the
parameters and different lattices (see
e.g. Fig.~\ref{fig_2d_rhombic}).

\subsection*{Numerical investigation} In the final Section \ref{sec-numerics} of the paper we complement the
analytical results with a numerical investigation on optimal lattices and charge
distributions for both the inverse power law and the Gaussian case. We have
  a rather complete picture, both analytically and numerically, of the solution
  of our minimization problem in two dimensions. In higher dimension, the
  numerics are more difficult to perform and we only compare the values of the
  rock-salt structures with other specific lattices which have the largest
  symmetry groups and which are also known as ``density stable" for radially
  symmetric interaction (see \cite{LBMorse}), i.e. the only possible lattices
  that are critical points in $\mathcal{L}(V)$ for $E_{f,0}$ in an open interval
  of volumes $V$. They are also the usual minimizers of $E_{f,0}$ in
  $\mathcal{L}$, and we expect these two properties to still hold in general
  for $E_{f_1,f_2}$. In particular, we consider the case of dimensions $2$,
  $3$ and $8$ (see \cite{CKMRV2Theta}).

  %\medskip

  We first note that, for a stable system, the attractive interaction between different charges related to $f_2$ needs to have a higher decay compared to the purely repulsive interaction related to $f_1$. In our setting, this amounts to the assumption $p > q$ (respectively $\alpha > \beta$). If $p-q$ (respectively $\beta-\alpha$) is positive, but sufficiently small, then the rock-salt structure is not optimal and the energy does not have a minimizer. Examples are given in Fig.~\ref{fig_Gaussip_nonopt} and Fig.~\ref{fig_ipnonopt}. In the following, we hence consider the situation
  when $p - q$ and $\beta-\alpha$ are large enough. In this situation, in all
  considered cases, the rock-salt structure seems to be optimal. More precisely,
  we consider the following cases:

  \medskip

  \textit{Dimension $d =2$:} The minimizer among all orthorhombic lattices and
  all values of $V$ is a rock-salt structure. This is illustrated by a plot over the fundamental
    domain (see Fig.~\ref{fig_square_rect}). Among all lattices with
  alternating charges, the minimizer at fixed $V>0$ exhibits a phase transition
  of the type ``triangular - rhombic - square - rectangular" as $V$ increases
  (see Fig.~\ref{fig_transition}). This was already observed for Lennard-Jones
  type interaction \cite{Beterloc,SamajTravenecLJ}, Morse type interaction
  \cite{LBMorse}, 3-block copolymers \cite{LuoChenWei} and two-component
  Bose-Einstein Condensates \cite{Mueller:2002aa}. Furthermore, the global
  minimum of $E_{f_1,f_2}^\pm$ among all lattices and all values of $V$ is a
  rock-salt structure (see Fig.~\ref{fig_2d_rhombic} and
  Fig.~\ref{fig_2d_p4q3inD}). The same optimality holds when comparing the
  rock-salt structure to $(\mathsf{A}_2,\varphi_\tinyvarhexagon)$, the triangular lattice with its honeycomb-like, energetically minimal distribution of charges.

  \medskip
  
  \textit{Dimensions $d\in \{3,8\}$:}  For $d=3$, by comparing the (lowest) energy of the rock-salt structure to FCC lattices and BCC lattices with alternating and optimal charge distributions (using the general formula proved in Thm.~\ref{thm-minenergyzetahomog}), we find out that the minimal energy among these lattices is obtained by the rock-salt structure (see Table \ref{table-3d} and Fig.~\ref{fig_3d_special}).
  
  For $d = 8$, among lattices with alternate distribution of charges, the cubic lattice with lowest energy has a lower energy than the $\mathsf{E}_8$ lattice with lowest energy (see Thm.~\ref{thm-minenergyzetahomog} and Table \ref{table-8d24d}).

  \medskip
  
  These results suggest that the rock-salt structure is the most promising
  candidate for the presented energy minimization problem for $d \in \{2,3,8\}$,
  for, both, inverse power law and Gaussian potentials, if the distance between
  the parameters is large enough. More generally, this suggests the following
  conjecture:
\begin{Conjecture}[Minimality of the rock-salt structure]\label{con-conj1}
	Let $d\geq 1$. Then there exist $\delta_0, \delta_1 > 0$ (depending only on $d$) such that the global minimizer of $E_{f_1,f_2}$, defined by \eqref{eq-mainenergyintro}, is of the form $\big(\mathcal{V}^{\frac{1}{d}}\Z^d,\varphi_\pm\big)$ for some $\mathcal{V}>0$ if either one of the following conditions is satisfied:
	\begin{enumerate}
		\item $f_1$, $f_2$ are given by \eqref{eq-deff1f2invpower} for $p,q>d/2$ with $p-q > \delta_0$.
		\item $f_1$, $f_2$ are given by \eqref{eq-deff1f2Gaussians} with $\beta-\alpha>\delta_1$.
	\end{enumerate}
	Moreover, if condition (i) holds, then $\mathcal{V}=V_L$, where $V_L$ is given by \eqref{def:VL}.
\end{Conjecture}
Our calculations suggest that the critical values of $\delta_0,\delta_1$ in dimensions $d\in \{2,3\}$
are rather small. In general, these values should depend
on the dimension.  More generally, we believe that Conjecture \ref{con-conj1}
also holds for completely monotone potentials $f_1$, $f_2$ as long as $f_1-f_2$
is a one-well potential. This has also been conjectured in \cite{Beterloc} for
the ``one type of particles" problem in dimension $2$.

For $d=1$, two of the authors already derived similar crystallization results
for systems with alternating charge distribution and two types of interactions
\cite{Crystbinary1d}. In particular, the optimality of the one-dimensional
rock-salt structure has been rigorously shown there. Furthermore, for
$d=2$, Friedrich and Kreutz \cite{FrieKreutSquare} recently proved the
optimality of a subset of the rock-salt structure (i.e., a finite
crystallization result) for short-range interactions $f_1$, $f_2$ among charges
of the form $\pm 1$. 

\subsection{Setting}\label{sec-setting}
We will now clarify the notation and introduce the integral concepts of this
work. We next introduce the notion of a lattice (sometimes also called ``Bravais lattice"). We remark that, in this work, vectors are understood as row vectors.
\begin{definition}[Lattices] \label{def-lattice}
Let $d \geq 1$.
\begin{enumerate}
	\item A lattice in $\Rd$ is a set of the form
	\begin{align} \label{eq-lattice}
		L \ = \ \bigoplus_{i=1}^d \Z u_i \qquad \text{ for some basis $\{u_1,...,u_d\}$ of $\R^d$.  }
     \end{align}
     The set of all lattices in $\R^d$ is denoted by $\mathcal{L}$ and the
     subset of lattices with fixed volume $V = |\det(u_1,...,u_d)|$ is denoted
     by $\mathcal L(V)$. The inverse volume $V^{-1}$ is also called the
       density of the lattice.   
     \item An orthorhombic lattice is a lattice of the form \eqref{eq-lattice} which can be represented by an orthogonal basis. We denote the set of orthorhombic lattices by $\mathcal{Q}$ and write $\mathcal{Q}(V) := \mathcal{Q} \cap \mathcal{L}(V)$.
         
	\item The triangular lattice $\mathsf{A}_2\in \mathcal{L}(1)$ is defined by $\mathsf{A}_2:= \sqrt{\frac{2}{\sqrt{3}}} \left[\Z(1,0)\oplus \Z(\frac{1}{2}, \frac{\sqrt{3}}{2})\right]$.
  \end{enumerate}
	We will also write $\mathcal{A} \subset \R^d$ for the set of vectors $a = \{a_1,...,a_d\} \in (0,\infty)^d$ such that $\prod_{i=1}^d a_i=1$ and write $\mathbb{1} :=(1,1,...,1)\in \mathcal A$. For $a \in \mathcal A$, we use the notation $L_a \in \mathcal Q(1)$ for an orthorhombic lattice of the form \eqref{eq-lattice} with $u_i = a_i e_i$.
\end{definition}

 We note that
%a lattice can equivalently be characterized as a discrete
%co-compact subgroup of $\Rd$.
any lattice $L \in \mathcal{L}(V)$ can be written
as $L = V^{1/d} \widetilde{L}$ for some $\widetilde{L} \in
\mathcal{L}(1)$. In our notation we have $L_{\mathbb{1}}=\Z^d$. We
  remark that in the crystallographic literature, orthorhombic lattices are
usually defined by the fact that their unit cell is cuboidal and all the lengths
$|u_i|$, $1\leq i\leq d$ are different (see e.g.,
\cite[Sect.~4.2.2.4]{NaumannBook}). However, we include the situation where
some or all the lengths $|u_i|$ are the same. We note that the choice of the
basis is non-unique and that e.g., an orthorhombic lattice can be represented with a basis whose elements are not orthogonal.
\begin{remark}[Two-dimensional lattices]\label{rmk:2dlattices}
	We recall that any two-dimensional lattice $L \subset \R^2$ of unit volume can be written in the form
	\begin{align}
		L \ = \ L(x,y) \ := \ \tfrac{1}{\sqrt{y}} \left[\Z (1,0) \oplus \Z (x,y)\right] \qquad \text{ with } (x,y) \in \mathcal{D}.
	\end{align}
	The pair $(x,y)$ is uniquely determined in the (right half) fundamental domain $\mathcal{D} \subset \R^2$;
	\begin{align} \label{def:Dlattices}
		\mathcal{D} \ := \ \left\{(x,y)\in \left[0, \tfrac 12 \right]\times (0,\infty) \ : \ x^2+y^2\geq 1\right\}.
	\end{align}
	In this setting, the triangular lattice $\mathsf{A}_2$ is represented by $\left(1/2,\sqrt{3}/2 \right)$ and the square lattice $\Z^2$ by $(0,1)$. Furthermore, rhombic lattices are characterized by $x^2+y^2=1$, $x \in [0, \tfrac{1}{2})$ and $(\tfrac{1}{2}, y) \in \mathcal{D}$ with $y \geq \tfrac{\sqrt{3}}{2}$. The name originates from the fact that, after a possible change of basis, the spanning vectors have equal length. The rectangular lattices $L_a$, $a=(y,y^{-1})$, are represented by $(0,y) \in \mathcal{D}$, $y \geq 1$.
	\flushright{$\diamond$}
\end{remark}

We recall the notion of charged lattices as defined in
\cite[Def.~1.1]{BeterminKnuepfer-preprint} and based on Born's paper
\cite{Born-1921}.  In a charged lattice, each $p \in L$ is assigned a charge
$\varphi(p)$. Even though the term ``charge" originally refers to the ionic
compounds setting, and can be understood as ``electric charge", it is important
for the reader to keep in mind that any notion of signed ``weight" can replace it.
\begin{definition}[Charged lattices] \label{def-chargedlattice}
	Let $d \geq 1$. For $L \in \mathcal L$ of the form \eqref{eq-lattice} and $N \in \N$, $\Lambda_N(L)$ is the set of $N$--periodic charge distributions $\varphi:L\mapsto \R$  such that		
	\begin{enumerate}
		\item $\displaystyle \varphi(x+Nu_i) = \varphi(x)$ for all $x \in L$ and $i \in \{ 1, \ldots, d \}$
		\vspace{0.7ex}
		\item $\varphi(0)>0$\label{charge_unique}
		\vspace{0.7ex}
		\item $\displaystyle \sum_{x\in K_N(L)} \varphi(x)^2 = N^d$,\label{charge_N}
	\end{enumerate}
	where the $N$-periodicity cube of $L$ is defined by
	\begin{align}\label{def-KN}
		K_N(L) \ := \ \Big\{x=\sum_{i=1}^d m_i u_i\in  L \ : \  m_i\in \{0, \ldots, N-1\}   \text{ for all } i \in \{1, \ldots, d\} \Big\}.
	\end{align}
	A charged lattice in $\R^d$ is a pair $(L,\varphi)$ consisting of a
        lattice $L \in \mathcal{L}$ and a (periodic) charge distribution
        $\varphi \in \Lambda_N(L)$ for some $N \in \N$.

        The alternate charge distribution $\varphi_\pm\in \Lambda_2(L)$ for $p = \sum_{i=1}^d m_i u_i$ is defined by
%        $\varphi_\pm (p)$ :=
%          $+1$ if $\sum_{i=1}^d m_i$ is even and $\varphi_\pm (p)$ := $-1$
%          else.
	\begin{align}
		\varphi_\pm(p) =
		\begin{cases}
			+ 1, & \sum_{i=1}^d m_i \equiv 0 \mod 2\\
			-1, & else
		\end{cases}
		.
	\end{align}
        %\hkrm{On the triangular lattice
        %  $\mathsf{A}_2$, the honeycomb-like charge distribution is defined by
        %  $\varphi_{\tinyvarhexagon}( \sqrt{\tfrac{2}{\sqrt{3}}}
        %  [m(1,0)+n(\tfrac{1}{2},\tfrac{\sqrt{3}}{2})]) := \sqrt{2}\cos
        %  (\tfrac{2\pi}{3}(2m+n))$.}

	The optimal, honeycomb-like, charge distribution $\varphi_\tinyvarhexagon\in \Lambda_3(L)$ for two-dimensional lattices is defined by
%	$\varphi_{\tinyvarhexagon}(p) := \sqrt{2}$ if
%          $\sum_{i=1}^2 (-1)^i m_i$ is divisible by $3$ and
%          $\varphi_{\tinyvarhexagon}(p) := -1/\sqrt{2}$ else.
	\begin{align}
		\varphi_{\tinyvarhexagon}(p) =
		\begin{cases}
			\sqrt{2}, & m_2-m_1 \equiv 0 \mod 3\\
          -1/\sqrt{2}, & else
		\end{cases}
		.
	\end{align}
\end{definition}

Assumption \ref{charge_unique} ensures the uniqueness of the optimal charge distribution, in particular, flipping all the charges of $\varphi_\pm$ does not have any effect. Assumption \ref{charge_N} is at the same time natural and technical. Indeed, we need to fix the total amount of charges on the periodicity cell, otherwise the problems under consideration do not have solutions. Also, it is widely used in the discrete Fourier method of \cite{BeterminKnuepfer-preprint}. In the definition of $\varphi_{\tinyvarhexagon}$, we corrected a typo in \cite[Thm.~2.6]{BeterminKnuepfer-preprint} where a factor $2$ is missing. For an illustration of $\varphi_\pm$ and $\varphi_\tinyvarhexagon$ see again Fig.~\ref{fig_opt_charge}.

Next, we introduce interaction potentials and the resulting lattice energies. Since our ``charges" are not necessarily ``electrical charges", they can interact through a potential which is not necessarily Coulombian. Also, for special choices of the interaction potential, we define special lattice functions.
\begin{definition}[Spaces of potentials]\label{def-potentials}
	Let $d\geq 1$. We say that $f\in \mathcal{F}$ if there exists a Borel measure $\mu_f$ on $(0,\infty)$ such that, for all $r>0$,
	\begin{align}
		f(r) \ = \ \int_0^\infty e^{-rt}d\mu_f(t)
	\end{align}
	and if $f(r)=O(r^{-s})$ as $r\to +\infty$, for some $s>d/2$.  If $\mu_f \geq 0$, we say that $f\in \mathcal{S}$.
\end{definition}

We note that $\mathcal{S}$ is just the class of completely monotone functions with sufficient algebraic decay so that the corresponding interaction potential $x \mapsto f(|x|^2)$ is integrable in $\R^d \backslash B_\eta(0)$ for arbitrary small $\eta>0$. Functions in $\mathcal{F}$ are those functions which can be written as the difference of two functions in $\mathcal{S}$. We recall the defining formula \eqref{eq-mainenergyintro} of the energy.
\begin{definition}[Energy]\label{def-energy}
	For any $f_1,f_2\in \mathcal{F}$, $N\in \N$ and $(L,\varphi)\in \mathcal{L}\times \Lambda_N(\mathcal{L})$, we define
	\begin{align}\label{eq-mainenergy}
		E_{f_1,f_2}[L,\varphi] \ := \ \sum_{p\in L\backslash \{0\}} f_1(|p|^2)+ \frac{1}{N^d}\sum_{p\in L\backslash \{0\}} \sum_{x\in K_N(L)} \varphi(x)\varphi(x+p)f_2(|p|^2).
	\end{align}
	Furthermore, for any $f_1,f_2\in \mathcal{F}$ and any lattice $L\in \mathcal{L}$, we define
	\begin{align}\label{eq-defaltenergy}
          E_{f_1,f_2}^\pm[L]\ := \ %
          E_{f_1,f_2}[L,\varphi_\pm] \ %
          = \ \sum_{p\in L\backslash \{0\}} f_1(|p|^2)+\sum_{p\in L\backslash \{0\}} \varphi_\pm(p) f_2(|p|^2).
	\end{align}
\end{definition}
Among all admissible interaction potentials, the inverse power laws
$r \mapsto r^{-s/2}$ and the exponential functions $r \mapsto e^{-\pi \alpha r}$
will play a special role in our study (see
Fig.~\ref{fig_intercation_function}). To the latter, we will also refer to as
the Gaussian case as we look at the potential function of squared distance,
i.e., $r \mapsto f(r^2)$, which then yields a Gaussian of the form $e^{- \pi \alpha r^2}$.

\begin{definition}[Epstein zeta functions and lattice theta functions] \label{def-epstein} \ 
	\begin{enumerate}
        \item The Epstein zeta function $\zeta_L$ and the alternating Epstein
          zeta function $\zeta_L^\pm$ of a lattice $L\in \mathcal{L}$ for
          $s>d$ are defined by
		\begin{align} \label{eq_alternate_EZ}
			\zeta_L(s) \ := \ \sum_{p\in L\backslash \{0\}} \frac{1}{|p|^s}, 
			\qquad \textnormal{ and } \qquad
			\zeta_L^\pm(s) \ := \ \sum_{p\in L\backslash \{0\}} \frac{\varphi_\pm(p)}{|p|^s}.
		\end{align}
		
              \item The lattice theta function $\theta_L$ and the alternating
                lattice theta function $\theta_L^\pm$ of $L\in \mathcal{L}$  for $\alpha > 0$ are
                 defined by
		\begin{align}\label{eq:theta}
                  \theta_L(\alpha) \ := \ %
                  \sum_{p\in L} e^{-\pi \alpha |p|^2}
			\qquad \textnormal{ and } \qquad
                  \theta_L^\pm(\alpha) \ %
                  := \ \sum_{p\in L} \varphi_\pm (p) e^{-\pi \alpha |p|^2}.
		\end{align}
		If $d=1$, these theta functions are the classical $\theta_3$- and $\theta_4$-function defined by
		\begin{align}\label{deftheta34}
			\theta_3(\alpha) \ := \ 
	%\hk{\theta_{\Z} (\alpha)}
	\sum_{k\in \Z} e^{-\pi \alpha k^2}
			\qquad \textnormal{ and } \qquad 	
			\theta_4(\alpha) \ := \
	%\hk{\theta^{\pm{}}_{L}(\alpha)}
	\sum_{k\in \Z} (-1)^k  e^{-\pi \alpha k^2}.
		\end{align}
	\end{enumerate}
\end{definition}

\begin{remark}[The non-absolutely summable case]\label{rmk-nonintegrablecase}
	The results of this paper still hold when $f_2$ is not assumed to decay sufficiently fast at infinity. In this case, a common way to define the second term of $E_{f_1,f_2}$ is to use the Ewald summation method, using, for example, a Gaussian convergence factor (see e.g., \cite{Ewald1,Ewaldpolytropic,SaffLongRange}). This method has been used by two of the authors in \cite{BeterminKnuepfer-preprint} and a definition of $E_{f_1,f_2}$ for such $f_2$ and $f_1\in \mathcal{S}$ would be 
\begin{align}
	E_{f_1,f_2}[L,\varphi] \ := \ \sum_{p\in L\backslash \{0\}} f_1(|p|^2)+ \lim_{\eta \to 0} \frac{1}{N^d}\sum_{p\in L\backslash \{0\}} \sum_{x\in K_N(L)} \varphi(x)\varphi(x+p)f_2(|p|^2)e^{-\eta |p|^2}.
\end{align}
This definition coincides with \eqref{eq-mainenergy} when $f_2\in \mathcal{S}$ and the optimality of the alternate charge distribution $\varphi_\pm$ for $\varphi\mapsto E_{f_1,f_2}[L,\varphi]$ when $L\in \mathcal{Q}$ still holds, if $\mu_{f_2}$ is a non-negative Borel measure. This is a consequence of the work done in \cite{BeterminKnuepfer-preprint}. By using the Ewald summation method, we can derive an expression of $E_{f_1,f_2}$ involving superpositions of Gaussians. The maximality of the cubic lattice for $E_{0,f_2}^\pm$ still holds, as well as the criticality of the cubic lattice proved in Thm.~\ref{thm-maincriticZd}. Since all the properties of the alternate Epstein zeta function -- in particular its homogeneity -- stay true for $s\leq d/2$, Thm.~\ref{thm-mainInvpowerlaws}, \ref{thm-local2dInvPowLaw} and \ref{thm-minenergyzetahomog} also hold in this case. Then, the analytic continuation of the general Epstein zeta function gives the same expression of the energy (again as a superposition of Gaussians).

This non-integrable case is of importance. In particular it is interesting in dimensions $d \geq 3$ for describing a Coulomb interaction $f_2(r^2)=r^{2-d}$ for charged particles. This is the commonly used potential for describing the interaction in ionic crystals where the charges are really understood as ``electric charges".
\flushright{$\diamond$}
\end{remark}

\section{Proof of theorems}
%\section{\hk{Optimality of the cubic lattice}}\label{sec-cubicCM}
In this section, we present the proofs for the theorems stated in the previous section.

\subsection{Criticality of the $d$-dimensional rock-salt structure -- Proof of Thm.~\ref{thm-maincriticZd}}\label{subsec:design}

Thm.~\ref{thm-maincriticZd} states that the $d$-dimensional rock-salt structure
$(\Z^d, \varphi_\pm)$ (where $\varphi_\pm$ is defined in
Def.~\ref{def-chargedlattice}) is a critical point in the space of unit volume
charged lattices of the form
$\left\{(L,\varphi_\pm) :L\in \mathcal{L}(1)\right\}$. We prove
Thm.~\ref{thm-maincriticZd} by using a result on 2-designs from
\cite{Coulangeon:2010uq}.
	
\begin{proof}[Proof of Thm.~\ref{thm-maincriticZd}]
 By a scaling argument we may assume
  $V=1$. Denoting the sublattices corresponding to the positive, respectively
  negative charges by $L_\pm \subset \Z^d$, we have
  \begin{align}\label{eq-Zddesign}
    E_{f_1,f_2}^\pm[L] = \sum_{p\in L\backslash \{0\}} f_1(|p|^2)+ \sum_{p\in L_+\backslash\{0\}} f_2(|p|^2) - \sum_{p\in L_-} f_2(|p|^2).
  \end{align}
  We recall that, for given $r > 0$, a layer $D$ of a lattice is the set
  of points in the lattice with $p \in S_r:=\partial B_r(0)$. A layer is called
  a $t$-design, if for any polynomial $P : \R^d \to \R$ of degree up to
  $t \in \N$, we have
  \begin{align}
    \frac{1}{|S_r|}\int_{S_r} P(x_1,...,x_d)dx = 
    \frac{1}{\sharp D}\sum_{x\in D} P(x_1,...,x_d).
  \end{align}
  We use the following result from \cite[Thm.~4.4(1)]{Coulangeon:2010uq}: If
  every layer $D\subset S_r$ of a lattice $L\in \mathcal{L}(1)$ is a $2$-design,
  then $L$ is a critical point of $L \mapsto E_{f,0}$ in $\mathcal{L}(1)$.  By
  \cite[Sect.~4]{CoulLazzarini}, we know that all the layers of $\Z^d$ are
  $2$-designs.  For any $n \in \Z^d$, by construction the charge at the
    point $p=\sum_{i=1}^d n_i e_i$ of the rock-salt structure is $+1$ if and
    only if $\sum_{i=1}^d n_i\in 2\Z$. This is equivalent to the fact that
    $|n|^2 = (\sum_{i=1}^d n_i)^2 -2\sum_{i\neq j} n_i n_j \in 2\N$. This
    implies that all the points of $\Z^d$ at distance $|n|^2$ to the origin have
    the same charge.  Therefore, all layers of $L_+$ (resp.~$L_-$) are also
  2-designs.  Thus, by \cite[Thm.~4.4.(1)]{Coulangeon:2010uq}, $L_+$,
  respectively $L_-$, is a critical point of the second, respectively third,
  term of the energy \eqref{eq-Zddesign}. Since $\Z^d$ is also a critical point
  of the first term of the energy (by the same argument), the proof is complete.
\end{proof}
The following remark is used in the proof of our next result.

% (Prop.~\ref{thm-mainsmallepsilon}):

\begin{remark}[Strict optimality of the cubic lattice for completely monotone kernels]\label{subsec:optiZdCM}
  The optimality of cubic lattices among the smaller class of
  orthorhombic lattices has been studied in \cite[Thm.~2]{Mont} for the lattice
  theta function and in \cite[Thm.~4]{BeterminPetrache} and
  \cite[Thm.~2.2]{Faulhuber:2016aa} for the alternate lattice theta
  function. The key point is that, for all $L_a\in \mathcal{Q}(1)$,
  $a \in \mathcal{A}$,
\begin{align}\label{eq:producttheta}
	\theta_{L_a}(\alpha) = \sum_{p \in L_a} e^{-\pi \alpha |p|^2} = \prod_{i=1}^d \theta_3(a_i^2 \alpha), \quad
	\theta_{L_a}^\pm(\alpha) = \sum_{p \in L_a} \varphi_\pm(p) e^{-\pi \alpha |p|^2} = \prod_{i=1}^d \theta_4(a_i^2 \alpha).
\end{align}
Then, using these product representations, it has been proved in \cite{Mont,BeterminPetrache,Faulhuber:2016aa} that, for all
$\alpha>0$, $\Z^d$ is the unique (strict) minimizer (resp. maximizer) of
$L\mapsto \theta_L(\alpha)$ (resp.~$L\mapsto \theta_L^\pm(\alpha)$) in
$\mathcal{Q}(1)$.  An important consequence is that the same result holds when
$f\in \mathcal{S}$, for the lattice energies
\begin{align}\label{eq-Egintegral}
		E_{f,0}^\pm[L_a] \ = \ \sum_{p\in L_a\backslash \{0\}} f(|p|^2)=\int_0^\infty \left(\prod_{i=1}^d \theta_3\left(\frac{a_i^2  t}{\pi}\right)-1\right)d\mu_f(t),
	\end{align}
	and
	\begin{align}\label{eq-Ehpmintegral}
		E^\pm_{0,f}[L_a] \ =\ \sum_{p\in L_a\backslash \{0\}}\varphi_{\pm}(p) f(|p|^2) \ = \ \int_0^\infty \left(\prod_{i=1}^d \theta_4\left(\frac{a_i^2  t}{\pi}\right)-1\right)d\mu_f(t),
	\end{align}
	where $\Z^d$ minimizes \eqref{eq-Egintegral} (resp.~maximizes \eqref{eq-Ehpmintegral}).  Furthermore, the Hessian of $E_{f,0}^\pm$ (resp. $E_{0,f}^\pm$) at $\Z^d$ is positive (resp. negative) definite. This result was already shown for the Epstein zeta function by Lim and Teo \cite{LimTeo} in the case of one type of particles \eqref{eq-Egintegral}.
	\flushright{$\diamond$}
\end{remark}

\subsection{The inverse power law case -- Proofs of Thm.~\ref{thm-mainInvpowerlaws}, \ref{thm-local2dInvPowLaw} and \ref{thm-minenergyzetahomog}}\label{sec-invpowerlaws}
In this section, we consider potentials of the form
\eqref{eq-deff1f2invpower}, i.e.
  \begin{align}
    f_1(r) \ = \ \frac{1}{r^p} \qquad \text{ and } \qquad %
     f_2(r) \ = \ \frac{1}{r^q}, \quad\text{where $p>q>\frac{d}{2}$},
  \end{align}
  when the energy can be
    written in terms of Epstein zeta functions, i.e.
    \begin{align}\label{eq:introalternatezeta}
      E_{f_1,f_2}^\pm[L] \ = \ \zeta_L(2p) + \zeta_L^\pm(2q), \qquad L\in \mathcal{L}.
    \end{align}
    Employing the homogeneity of the potentials, $E_{f_1,f_2}^\pm$ can then
    be expressed in the form
\begin{align}\label{eq-FVaPowerLaw}
  E_{f_1,f_2}^\pm[V^{\frac{1}{d}}L] \ = \ %
  \frac{1}{V^{\frac{2p}{d}}} \left(\zeta_{L}(2p) +V^{\frac{2(p-q)}{d}}\zeta_{L}^\pm(2q) \right),\quad L\in \mathcal{L}(1), V>0.
      \end{align}
      Thm.~\ref{thm-mainInvpowerlaws} states the minimality of the cubic
        lattice at high density and its non-minimality at low density. Indeed,
        formula \eqref{eq-FVaPowerLaw} suggests that for $V$ sufficiently
      small, $L\mapsto E_{f_1,f_2}^\pm[V^{\frac{1}{d}}L]$ should be minimized by
      the minimizer of $L\mapsto \zeta_L(2p)$. For $V$ very large,
      $E_{f_1,f_2}^\pm[V^{\frac{1}{d}}L]$ should not be minimized by the same
      lattice, since $L\mapsto \zeta_L^\pm(2q)$ does not have a minimizer as a
      simple consequence of \cite[Prop.~1.3]{BeterminPetrache} by using the
      Poisson Summation Formula (see, e.g., \cite[Cor.~2.4]{SteinWeiss}). In the
      following proof, we give a corresponding, rigorous argument. 

\begin{proof}[Proof of Thm.~\ref{thm-mainInvpowerlaws}]
As recalled at the beginning of Section 2, we already know (see \cite{BeterminKnuepfer-preprint}) that the alternate distribution of charges $\varphi_\pm$ is the unique minimizer of $\varphi\mapsto E_{f_1,f_2}[L,\varphi]$ in $\mathcal{Q}(V)$, for any $V>0$ and any $f_1,f_2\in \mathcal{S}$. It is thus sufficient to find the minimizer of $E_{f_1,f_2}^\pm$ in $\mathcal{Q}(V)$.  We first consider the case of small $V$, where we claim that the lattice
  $V^{\frac{1}{d}}\Z^d$ is energetically optimal. We note that, for any $V>0$
  and any $L\in \mathcal{Q}(1)$ we have
	\begin{align}
          E_{f_1,f_2}^\pm[V^{\frac{1}{d}}L] \ %
          \stackrel{\eqref{eq-FVaPowerLaw}}= \ %
          \frac{1}{V^{\frac{2p}{d}}}  E_{f_1,\varepsilon f_2}^{\pm}[L],
          \qquad \qquad%
          \text{where $\varepsilon \ := \ V^{\frac{2(p-q)}{d}}$}.
	\end{align}
Let us define, for any $n\in \Z^d$ and any $a\in \mathcal{A}$, $R_{a}^2:= \sum_{j=1}^d n_j^2 a_j^2$ (we voluntarily omit the $n$-dependence) and 
\begin{align} \label{f2-con} %
		\gamma(a) \ := \ \Big( \sum_{n\in \Z^d\backslash \{0\}}(f_1(R_{a}^2)-f_1(R_{\mathbb 1}^2)) \Big) \Big/
		\Big( \sum_{n\in \Z^d\backslash \{0\}}(-1)^{|n|}  (f_2(R_{\mathbb 1}^2)-f_2(R_{a}^2) ) \Big).
	\end{align}
We now remark that, for all $\mathcal{A}\ni a\neq \mathbb{1}$ and any $L_a\in \mathcal{Q}(1)$, $E_{f_1,\varepsilon f_2}^{\pm}[L]>E_{f_1,\varepsilon f_2}^{\pm}[\Z^d]$ i.e.
    \begin{align}
      \sum_{n\in \Z^d\backslash \{0\}}\left(f_1(R_{a}^2)+\varepsilon(-1)^{|n|}f_2(R_{a}^2)\right)>\sum_{n\in \Z^d\backslash \{0\}}\left(f_1(R_{\mathbb 1}^2)+\varepsilon(-1)^{|n|}f_2(R_{\mathbb 1}^2)\right),
    \end{align}
  if and only if
    $\varepsilon < \varepsilon_0:=\inf_{\stackrel{a\in \mathcal{A}}{a\neq
        \mathbb{1}}} \gamma(a)$, the sign of the inequality being ensured by the strict minimality
        (resp. maximality) of $a=\mathbb{1}$ for $E_{f_1,0}^\pm$
        (resp. $E_{0,f_2}^\pm$) as explained in Rmk.~\ref{subsec:optiZdCM}. It
      remains to show that $\varepsilon_0>0$. For this, we note that by
      l'Hôpital's rule and the results stated in Rmk.~\ref{subsec:optiZdCM},
      we have
	\begin{align} 
          \eta \ := \  \inf_{\mathcal{A}\backslash \{ \mathbb{1}\}  \ni a^{(k)} \to \mathbb{1}}\liminf_{k \to \infty}\gamma(a^{(k)}) \ >  \ 0,
	\end{align}
        since the first derivatives with respect to $a$ at $a=\mathbb 1$ of,
        both, numerator and denominator in \eqref{f2-con} vanish and the second
        derivatives are strictly positive as a straightforward calculation
        shows. Since $\mathbb{1}$ is the unique zero of both numerator and
        denominator, there exists $\delta > 0$ such that
        $\gamma(a) > C_1(\delta) >0$ for some $C_1(\delta)$ whenever
        $|a - \mathbb 1| < \delta$. By the minimality of the square lattice and
        the continuity of $a \mapsto \gamma(a)$, we have
        $\gamma(a) > C_2(\delta,R)$ for any $R > 0$ and any $a \in \mathcal{A}$
        with $|a- \mathbb 1| \in (\delta,R)$. Since $f_2(r)=o(f_1(r))$ as
        $r\to 0$ and both functions go to $0$ at infinity, it follows that
        $\liminf_{R \to \infty} \inf_{|a-\mathbb 1| > R} \gamma(a) >
        0$.  Therefore, there exists
        $\varepsilon_0 > 0$ such that for all $0\leq \varepsilon<\varepsilon_0$,
        $\Z^d$ is the unique minimizer of
        $L\mapsto E_{f_1,\varepsilon f_2}^{\pm}[L]$ in
        $\mathcal{Q}(1)$. Expressed differently, there exists $V_0$ such that
        for $V<V_0$ we have that $V^{\frac{1}{d}} \Z^d$ is the unique minimizer
        of $E_{f_1,f_2}^{\pm}$ in $\mathcal{Q}(V)$.

        \medskip
        
	We will now prove the non-minimality of the cubic lattice at low density, i.e., for large $V$. By \eqref{eq-FVaPowerLaw}, we have, for any orthorhombic lattice $L_a\in \mathcal{Q}(1)$, $a\in \mathcal{A}$, and any $V>0$,
	\begin{align}
          V^{\frac{2p}{d}}\left[E_{f_1,f_2}^\pm[V^{\frac{1}{d}}L_a]-E_{f_1,f_2}^\pm[V^{\frac{1}{d}}\Z^d]  \right] \ %
          = \ \zeta_{L_a}(2p)-\zeta_{\Z^d}(2p)+V^{\frac{2(p-q)}{d}}\left(\zeta^\pm_{L_a}(2q)-\zeta^\pm_{\Z^d}(2q)\right).
	\end{align}
	By the
          optimality of $\Z^d$ for completely monotone kernels (see
          Rmk.~\ref{subsec:optiZdCM}), we know that for all $a\in \mathcal{A}\backslash \{\mathbb{1}\}$,
          $ \zeta_{L_a}(2p)-\zeta_{\Z^d}(2p)> 0$ and
          $\zeta^\pm_{L_a}(2q)-\zeta^\pm_{\Z^d}(2q)< 0$. It follows that $E_{f_1,f_2}^\pm[V^{\frac{1}{d}}L_a]-E_{f_1,f_2}^\pm[V^{\frac{1}{d}}\Z^d]\geq 0$ for any $a\in \mathcal{A}\backslash \{\mathbb{1}\}$ if and only if
	$$
	V\leq V_1:=\inf_{a\in \mathcal{A}\backslash \{\mathbb{1}\}} \overline{\gamma}(a),\qquad\qquad %
        \textnormal{where}\quad \overline{\gamma}(a):= \left(\frac{\zeta_{L_a}(2p)-\zeta_{\Z^d}(2p)}{\zeta^\pm_{\Z^d}(2q)-\zeta^\pm_{L_a}(2q)}\right)^{\frac{d}{2(p-q)}}.
	$$
	As explained in the previous proof, we know by l'Hôpital's rule that the
        above quotient $\overline{\gamma}(a)$ has a strictly positive limit as
        $a\to \mathbb{1}$ in $\mathcal{A}$. Furthermore, we notice that
        
        $$
        \liminf_{R\to \infty} \inf_{|a-\mathbb{1}|>R}\overline{\gamma}(a)>0
        $$ 
        by a simple application of the integral representations \eqref{eq-Egintegral}-\eqref{eq-Ehpmintegral} and the fact that 
        $$
        \liminf_{R\to \infty} \inf_{|a-\mathbb{1}|>R} \theta_{L_a}(t)/\theta^\pm_{L_a}(t)>0,\quad \forall t>0,
        $$ 
        as a straightforward application of the product representations \eqref{eq:producttheta} and the Poisson summation formula.
        Therefore, it follows by continuity of $a\mapsto \overline{\gamma}(a)$ in
        $\mathcal{A}\backslash \{\mathbb{1}\}$ that $V_1$ exists and is finite.
	\end{proof}

We will now give the proof of Thm.~\ref{thm-local2dInvPowLaw}.

\begin{proof}[Proof of Thm.~\ref{thm-local2dInvPowLaw}]
  Let $p > q > 1$. We recall that any orthorhombic (i.e., rectangular) lattice
  $L_a \subset \R^2$ in two dimensions can be written in the form
  $L_a=\Z\left(y,0 \right) \oplus \Z\left( 0,y^{-1} \right)$ for $y > 0$
  (cf. Rmk.~\ref{rmk:2dlattices}). By Thm.~\ref{thm-maincriticZd}, $V^{\frac{1}{2}} \Z^2$ is a critical point of $E_{f_1,f_2}^\pm$ in $\mathcal{Q}(V)$. For $a = (y^{-1}, y)$, let $L_a\in \mathcal{Q}(1)$. For fixed $V>0$, we compute the second derivative of $y \mapsto E_{f_1,f_2}^\pm[V^{\frac{1}{2}}L_a]$ and evaluate it at $y=1$. By using the homogeneity of $f_1,f_2$ we get (see also \cite{Beterloc})
	\begin{align}
	\frac{d^2}{dy^2} E_{f_1,f_2}^\pm[V^{\frac{1}{2}}L_a] \Bigg|_{y=1}=\frac{1}{V^p}S_1(p)+\frac{1}{V^q}S_2(q) \ =  \ \frac{1}{V^q}\left(\frac{1}{V^{p-q}}S_1(p)+S_2(q)  \right),
	\end{align}
   where, using the fact that $\displaystyle \sum_{(m,n) \in \Z^2\backslash \{0\}}\frac{m^4}{(m^2+n^2)^{p+2}}=\sum_{(m,n) \in \Z^2\backslash \{0\}}\frac{n^4}{(m^2+n^2)^{p+2}}$,
	\begin{align}
          & S_1(p) \ = \ 4p\sum_{(m,n) \in \Z^2\backslash \{0\}} \frac 1{{(m^2+n^2)^{p+2}}} \left( (p+1)(n^2-m^2)^2- 2m^4 -2m^2n^2 \right),\\
          & S_2(q)\ = \ 4q\sum_{(m,n) \in \Z^2\backslash \{0\}} \frac{(-1)^{m+n}}{(m^2+n^2)^{q+2}} \left( (q+1)(n^2-m^2)^2 - 2m^4  - 2m^2 n^2 \right).
	\end{align}
	We note that $S_1(p)>0$ and $S_2(q)<0$. This is due to the fact that
        $\theta_3 > 1$ and $\theta_4 < 1$ and by expressing $S_1$ and $S_2$ in
        terms of theta functions by \eqref{eq-Egintegral} and
        \eqref{eq-Ehpmintegral}, respectively. The statements of the theorem
        then follow with the definition
        \begin{align}\label{def:Vpq2d}
          V_{p,q} \ := \ \left( \frac{ S_1(p)}{-S_2(q)}\right)^{\frac{1}{p-q}}.
    \end{align}
\end{proof}
We note that $V_{p,q}$ is explicit and easily computable.  To conclude the
inverse power law results, we give the proof of
Thm.~\ref{thm-minenergyzetahomog}:
\begin{proof}[Proof of Thm.~\ref{thm-minenergyzetahomog}]
	We need to compute the minimum of our energy model among all the possible volumes, for a given $L\in \mathcal{L}(1)$, defined by
	\begin{align}
          \displaystyle V\mapsto  E_{f_1,f_2}^\pm[V^{\frac{1}{d}}L] \ %
          = \ \frac{\zeta_L(2p)}{V^{\frac{2p}{d}}}+\frac{\zeta_L^\pm(2q)}{V^{\frac{2q}{d}}}.
	\end{align}
	The aim is to compare the minimal energies among the dilates of $L$ for different lattices. This method has already been used in \cite[Section 5]{OptinonCM} for comparing the Lennard-Jones type energies of different lattices. We define
	\begin{align}\label{def:VL}
		V_L \ := \ \left(\frac{p \, \zeta_L(2p)}{-q \, \zeta_L^\pm(2q)}  \right)^{\frac{d}{2(p-q)}}.
	\end{align}
	
	The $\zeta_L^\pm(2q)\geq 0$ case is trivial. For the other case, we easily obtain
	\begin{align}
	  	\dfrac{\partial}{\partial V} E_{f_1,f_2}^\pm[V^{\frac{1}{d}}L] \ %
		= \ \dfrac{\partial}{\partial V} \left( \frac{\zeta_L(2p)}{V^{\frac{2p}{d}}}+\frac{\zeta_L^\pm(2q)}{V^{\frac{2q}{d}}} \right)
		= - \frac{2p \, \zeta_L(2p)}{d \, V^{\frac{2p}{d}+1}}-\frac{2q \, \zeta_L^\pm(2q)}{d \, V^{\frac{2q}{d}+1}}.
	\end{align}
	It follows that, since $\zeta_L^\pm(2q)<0$, 
	\begin{align}
	  	\dfrac{\partial}{\partial V} E_{f_1,f_2}^\pm[V^{\frac{1}{d}}L] \ \geq \ 0 \iff V \ \geq \ V_L.
	\end{align}
	For given $L$, it is therefore easy to check that the minimal energy is given by $\mathcal{E}_L^\pm$ as in \eqref{eq-minenergyL}.
\end{proof}

\subsection{The Gaussian case -- Proof of Thm.~\ref{thm-mainGaussians}}\label{sec-Gaussian}
In this section, we assume that $r \mapsto f_1(r^2)$ and
$r \mapsto  f_2(r^2)$ are Gaussian functions of the form
\begin{align}\label{eq_Gaussian}
	f_1(r^2) \ = \ e^{-\pi \beta r^2},\quad f_2(r^2) \ = \ e^{-\pi \alpha r^2}, \qquad \beta > \alpha >0.
\end{align}
In this case, the energy can be expressed in terms of the lattice theta
function $\theta_L$ and the alternate lattice theta function $\theta^\pm_L$ (see
Def.~\ref{def-epstein}) by
\begin{align}
	E_{f_1,f_2}^\pm[L] \ := \ \theta_L(\beta) +\theta_L^\pm(\alpha)-2.
\end{align}
Subtracting 2 in the above model comes from the fact that we exclude the origin
from the summation in $E_{f_1,f_2}^\pm$, but it is included in the definition of
the theta functions.  By rescaling lengths, it is enough to consider the
lattices of unit volume. The energy is then given, for any
$L_a\in \mathcal{Q}(1)$, $a\in \mathcal{A}$ and $V>0$, by
\begin{align}
	E_{f_1,f_2}^\pm[V^{\frac{1}{d}}L_a] \ 
	&= \ \prod_{i=1}^d \theta_3 \big( a_i^2 V^{\frac{2}{d}}\beta \big) + \prod_{i=1}^d \theta_4 \big( a_i^2 V^{\frac{2}{d}} \alpha \big) - 2 \
	=  \ \theta_{L_a} \big( V^{\frac{2}{d}} \beta \big) + \theta_{L_a}^\pm \big( V^{\frac{2}{d}} \alpha \big) - 2.
\end{align}
We will now prove Thm.~\ref{thm-mainGaussians}. The goal is to determine
(non-sharp) ranges for the parameters, such that the rock-salt structure either
minimizes the energy model $E_{f_1,f_2}^\pm$ or such that it is a local
maximizer of it. The main ingredient is the asymptotic behavior of Jacobi's
theta functions.

\begin{proof}[Proof of Thm.~\ref{thm-mainGaussians}]
  (i): We start to prove the non-minimality of the cubic lattice for
  $E_{f_1,f_2}^\pm$ when $V$ is large enough, in dimension 2. Let
  $L_a \in \mathcal{Q}(1)$, $a = (y,y^{-1})$, $y\geq 1$. It is convenient to
  write the energy in the form
	\begin{align}	
		E_{f_1,f_2}^\pm[V^{\frac{1}{2}}L_a] \ = \ f_{\beta,V}(y)+g_{\alpha,V}(y)-2,
	\end{align}
	where
	\begin{align}
          f_{\beta,V}(y) \ := \ %
          \theta_3(\beta V y)\theta_3\left(  \frac{\beta V}{y}\right), %
          \qquad %
          g_{\alpha,V}(y) \ := \ \theta_4(\alpha V y)\theta_4\left(\frac{\alpha V}{y} \right).
	\end{align}
	We note that, using Remark \ref{subsec:optiZdCM} which guarantees the fact that $g_{\alpha,V}(1)>g_{\alpha,V}(y)$ and $f_{\beta,V}(y)>f_{\beta,V}(1)$ for all $y>1$,
	\begin{align}\label{eq-equivFV}
		E_{f_1,f_2}^\pm[V^{\frac{1}{2}}L_a] \ > \ E_{f_1,f_2}^\pm[V^{\frac{1}{2}}\Z^2] \quad \text{$\forall L_a \in \mathcal Q(1)$}
		\quad \iff \quad 
		\inf_{y \neq 1} \frac{f_{\beta,V}(y)-f_{\beta,V}(1)}{g_{\alpha,V} (1)-g_{\alpha,V}(y)} > 1.
	\end{align}
	In particular, for $y \to 1$, by l'Hôpital's rule we get
	\begin{align}\label{eq_Gauss_asymptote}
          \lim_{y \to 1} \frac{f_{\beta,V}(y)-f_{\beta,V}(1)}{g_{\alpha,V} (1)-g_{\alpha,V}(y)}\ 
          = \ \frac{f_{\beta,V}''(1)}{-g_{\alpha,V}''(1)} \ 
          =  \ \frac{\beta^2}{\alpha^2} \, e^{-\pi V(\beta-\alpha)} +  o(e^{-\pi V(\beta-\alpha)})
	\end{align}
       for large $V \gg 1$. The asymptotic result in \eqref{eq_Gauss_asymptote} follows by a straightforward calculation by computing and estimating the derivatives of $\theta_3$ and $\theta_4$. Since $\beta-\alpha>0$, it follows that there exists
        $V_1$ such that
	\begin{align}
		\lim_{y \to 1} \frac{f_{\beta,V}(y)-f_{\beta,V}(1)}{g_{\alpha,V} (1)-g_{\alpha,V}(y)} \ < \ 1 \qquad \text{for all $V>V_1$.}
	\end{align}
	Combining \eqref{eq-equivFV}-\eqref{eq_Gauss_asymptote} and the fact
        that $f_{\beta,V}''(1)>0$ and $-g_{\alpha,V}''(1)>0$ (as a consequence of
        Rmk.~\ref{subsec:optiZdCM}), it follows that
        $\frac{d^2}{d y^2}E_{f_1,f_2}^\pm[V^{\frac{1}{2}}L_a]
        \big|_{y=1}=f_{\beta,V}''(1)+g_{\alpha,V}''(1) < 0$ for $V > V_1$. Therefore, the square lattice
          $V^{\frac 12} \Z^2$ is a strict local maximizer of $E_{f_1,f_2}^\pm$
          in $\mathcal{Q}(V)$ for all $V>V_1$.

	\medskip
        
        \textit{(ii):} We next prove the strict minimality of the cubic lattice
        for $E_{f_1,f_2}^\pm$ when $\alpha$ is small enough. We first remark that, for any $L_a\in \mathcal{Q}(1)$ where
        $a=(y,y^{-1})$, $y>0$, we have
	\begin{align}
          \frac{d^2}{d y^2}E_{f_1,f_2}^\pm[V^{\frac{1}{2}}L_a] \Big|_{y=1} \ %
          = \ \frac{d^2}{d y^2}\left[  \theta_3(\beta V y)\theta_3\left(\frac{\beta V}{y}  \right)\right]_{y=1} %
          + \frac{d^2}{d y^2}\left[  \theta_4(\alpha V y)\theta_4\left(\frac{\alpha V}{y}  \right)\right]_{y=1}.
	\end{align}
	The first term is strictly positive for any fixed $\beta$ and $V$ by the
        strict minimality of the cubic lattice at $y=1$ (see Rmk.~\ref{subsec:optiZdCM}). The second term converges to $0$ as
        $\alpha \to 0$. Indeed,
        \begin{align}
          \frac{d^2}{d y^2}\left[  \theta_4(\alpha V y)\theta_4\left(\frac{\alpha V}{y}  \right)\right]_{y=1} %
          = 2\alpha^2 V^2 \theta_4(\alpha V) \theta_4''(\alpha V)-2\alpha^2 V^2 \theta_4'(\alpha V)^2+2\alpha V \theta_4(\alpha V) \theta_4'(\alpha V).
        \end{align}
        The convergence to zero is then a simple consequence of the fact the
        $\theta_4$ function and its first two derivatives are bounded and
        continuous on $[0,\infty)$ (see e.g.~\cite{WhiWat69}). Hence, $y=1$ is a
        strict local minimum of $a\mapsto E_{f_1,f_2}^\pm[V^{\frac{1}{2}}L_a]$
        for all $\alpha<\alpha_0$ where $\alpha_0 = \alpha_0(\beta,V)$ is small
        enough. Furthermore, for any orthorhombic lattice $L_a$, there exists
        $\alpha_1=\alpha_1(a, \beta, V)>0$ such that
\begin{align}\label{eq:minorthoargument}
  E_{f_1,f_2}^\pm[V^{\frac{1}{2}}L_a]-E_{f_1,f_2}^\pm[V^{\frac{1}{2}}\Z^2] \ > \ 0 \qquad \forall \alpha \ < \ \alpha_1.
\end{align}
This is a direct consequence of the fact that the second term of the energy goes
to $0$ as $\alpha\to 0$ when $a$ and $V$ are fixed.

\medskip
    
Let $\alpha_0$ be such that $y=1$ is a strict local minimizer of
  $a\mapsto E_{f_1,f_2}^\pm[V^{\frac{1}{2}}L_a]$. Therefore, there exists
  $\eta > 0$ such that
  $ E_{f_1,f_2}^\pm[V^{\frac{1}{2}}L_a]-E_{f_1,f_2}^\pm[V^{\frac{1}{2}}\Z^2] >
  0$ for all $a=(y,y^{-1})$, $y\in I_\eta :=(1-\eta,1+\eta)\backslash \{1\}$. Furthermore, for any $a=(y,y^{-1})$
  where $y \in  \R_+\backslash I_\eta \cup \{1\}$, there exists
  $\alpha_1(y)$ such that \eqref{eq:minorthoargument} holds. By continuity of
  $\alpha\mapsto E_{f_1,f_2}^\pm[L]$ for any given $L\in \mathcal{L}$ and
  $\beta>0$, the fact that $\alpha_1(1+\eta)>0$ and $\alpha_1(1-\eta)>0$ and since
  \begin{align}
    \lim_{\alpha\to 0} E_{f_1,f_2}^\pm[V^{\frac{1}{2}}L_a]=\theta_3(\beta V y)\theta_3\left(\frac{\beta V}{y}  \right)=f_{\beta,V}(y)
  \end{align}
        is a strictly increasing function of $y$ for any $V,\beta>0$  which grows faster than $g_{\alpha,V}$ in $[y_0,\infty)$ for large $y_0>1+\eta$ (i.e. on $(0,y_0)$ for small $y_0$ by symmetry) and for $\alpha<\beta$ (see
        e.g. \cite{Faulhuber:2016aa}), it follows that
        $\alpha_2:=\inf_{y\in \R_+\backslash I_\eta\cup \{1\}}\alpha_1(y)>0$. Therefore,
        $E_{f_1,f_2}^\pm[V^{\frac{1}{2}}L_a]-E_{f_1,f_2}^\pm[V^{\frac{1}{2}}\Z^2]
        > 0$ for all $\alpha<\min \{\alpha_0,\alpha_2\}$ and all
        $a\in \mathcal{A}$. This concludes the proof of (ii).
        
        \medskip
        
	\textit{(iii):} The last point of the theorem is shown as in the
          proof of Thm.~\ref{thm-mainInvpowerlaws} for the special choice of
          Gaussian potentials $f_1(r^2) = e^{-\pi \beta r^2}$ and
          $f_2(r^2) = e^{-\pi \alpha r^2}$. We will omit the proof as it is more
          or less a repetition of the proof of the high density result of
          Thm.~\ref{thm-mainInvpowerlaws} and the facts that
          $\theta_3(t)\sim 1$ and $\theta_2(t):=t^{-1/2}\theta_4(t^{-1})\sim 0$ as $t\to +\infty$, after using the Poisson Summation formula. 
\end{proof}
We remark that for fixed $\beta > \alpha$, it is also possible to numerically
compute an upper bound for $V_1$ by determining
$\widetilde{V}_1:=\inf \{ V>0 : - f_{\beta,V}''(1) (g_{\alpha,V}''(1))^{-1}
<1\}$.
\begin{remark}[Connection with two-component Bose-Einstein Condensates]\label{rmk-BEC}
	Using the Poisson summation formula (see, e.g., \cite[Cor. 2.6]{SteinWeiss}), we can show that, for any $V > 0$ and $0 < \alpha < \beta$,  setting $t=\frac{\beta}{\alpha}$ and $s=\frac{1}{V^{\frac{2}{d}} \, \beta}$, we have
	\begin{align}\label{eq:energygaussian2}
			E_{f_1,f_2}^\pm[V^{\frac{1}{d}}L_a] \ = \ \frac{1}{s^{\frac{d}{2}}}\left\{\theta_{L_a}(s)+t^{\frac{d}{2}}\theta_{L_a+\frac{a}{2}}( ts)\right\}-2,
	\end{align}
	for all $L_a\in \mathcal{Q}(1)$, $a\in \mathcal{A}$. Therefore, the
        problem of minimizing $E_{f_1,f_2}^\pm$ can be related to the following
        two-component Bose-Einstein Condensates minimization problem originally
        described by Mueller and Ho in \cite{Mueller:2002aa} (see also
        \cite{ReviewvorticesBEC} for a review). They consider two 2-dimensional
        lattices of the same kind, shifted by a vector $z \in \R^2$ and then
        look for the minimizer in $\mathcal{L}(1)\times \R^2$ of the following
        energy
        $E_\delta(L,z) := \theta_L(1) + \delta \, \theta_{L+z}(1),\quad 0\leq
        \delta \leq 1$. Numerically, they observed that, as $\delta$ increases
        from $0$ to $1$, the minimizer $(L_\delta,z_\delta)$ of $E_\delta$ in
        $\mathcal{L}(1)\times \R^2$ exhibits a transition
        $(\mathsf{A}_2,b_{\mathsf{A}_2})$ $\rightarrow$
        $(\textnormal{Rhombic}, c_L)$ $\rightarrow$ $(\Z^2,c_{\Z^2})$
        $\rightarrow$ $(\textnormal{Rectangular}, c_L)$, where
        $b_{\mathsf{A}_2}$ is the barycenter of a primitive triangle of
        $\mathsf{A}_2$ and $c_L$ is the center of the unit cell of the
        respective lattice. This is precisely the type of lattice phase
        transition we have observed for the two-dimensional inverse power law
        and Gaussian cases (see Fig.~\ref{fig_transition}). However, for
        $E_\delta$ the minimizer passes from triangular to rhombic in a
        discontinuous way while our numerics suggest that the transition is
          continuous for our energy $E_{f_1,f_2}^\pm$.
	\flushright{$\diamond$}
\end{remark}

\section{Numerical investigation}\label{sec-numerics}

This section is devoted to numerics to investigate the global optimality of the rock-salt structure for
inverse power law and Gaussian interactions (see Conjecture
\ref{con-conj1}). We order the presentation of our results with respect to the dimension. We used Mathematica
\cite{Mathematica} to perform the numerics. 
  We have also checked that the results of this paper still
should hold for the three-dimensional Coulomb attraction case $q=1/2$.

\subsection{The two-dimensional case} As explained in Rmk.~\ref{rmk:2dlattices}, the indexing of lattices in $\mathcal{L}(V)$, by using the
domain $\mathcal{D}$, is well-understood in dimension $d=2$. When the potentials
$f_1,f_2$ are absolutely summable, it is then easy and fast (in terms of
computation time) to compute, illustrate and compare the energy values in
$\mathcal{L}(V)$.

\subsubsection*{\textbf{a) Minimization of $E_{f_1,f_2}^\pm$ in $\mathcal{Q}$.}}
Recall that any two-dimensional orthorhombic (also called rectangular) lattice
can be characterized by one parameter $y \geq 1$, describing the geometry, and
one parameter $V > 0$, fixing the volume $V$, i.e.
$V^{\tfrac{1}{2}} L_a = V^{\tfrac{1}{2}} \big(y^{-1} \Z \times y \Z \big)$ for
$L_a \in \mathcal{Q}(1)$. 
%For the numerics in the inverse power law case, we
%have chosen to fix $(p,q)=(4,3)$. 

In Fig.~\ref{fig_square_rect}, we have plotted the function $V\mapsto E_{f_1,f_2}^\pm[V^{\frac{1}{d}}L_a]$ for different values of $y$ and
fixed parameters $(p,q)=(4,3)$. The minimal energy is achieved by a square
lattice (i.e., $y=1$) which means that the minimum of $E_{f_1,f_2}^\pm$ in
$\mathcal{Q}$ is achieved by a rock-salt structure.
\begin{figure}[!htb]
	\includegraphics[width=.5\textwidth]{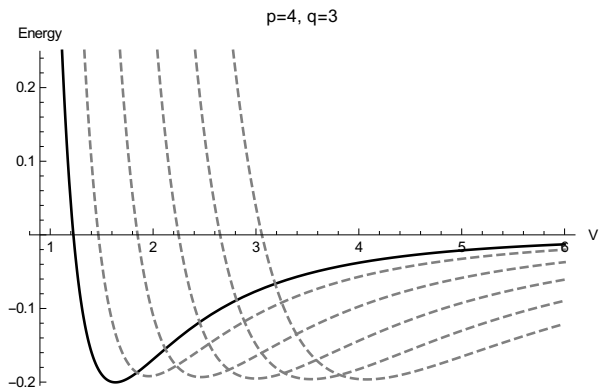}
	\caption{\footnotesize{Comparison of $V\mapsto E_{f_1,f_2}^\pm[V^{\frac{1}{2}}L_a]$ for different rectangular lattices $V^{1/2} \left(y^{-1} \Z \times y \Z \right)$ where $y\in \{1,1.4,1.8,2.2,2.6,3\}$, and inverse power laws $(f_1(r),f_2(r))=(r^{-4}, r^{-3})$. The solid line is the curve for the square lattice. Therefore, among structures with lattice in $\mathcal{Q}$, the energy $E_{f_1,f_2}^\pm$ seems to be minimized by a rock-salt structure.}}
	\label{fig_square_rect}
\end{figure}
      
Qualitatively, the same behavior (in accordance with
Thm.~\ref{thm-mainGaussians}), is observed in the Gaussian case for $\alpha = 1$, $\beta = 2$ (see Fig.~\ref{fig_Gauss_V}).  The critical value of the volume, for which the square lattice seems to be the \textit{global} minimizer in $\mathcal{Q}$, is $V_0\approx 0.464916$, as shown in Fig.~\ref{fig_Gauss_V} (a). Furthermore, in Fig.~\ref{fig_Gauss_V} (b), we have plotted the second derivative of $y \mapsto E_{f_1,f_2}[V^{\frac{1}{2}}L_a]$, evaluated at $y=1$, as a function of the volume $V$. Thus, we can check when the square lattice $V^{\frac{1}{2}}\Z^2$ is a local minimum or maximum of $E_{f_1,f_2}^\pm$. Again, it seems that there exists a critical value $V_1$ such that for all $V<V_1$ (resp. $V>V_1$) the square lattice is a local minimizer (resp. local maximizer) of $E_{f_1,f_2}^\pm$ in $\mathcal{Q}(V)$. Furthermore, numerical investigation show that if the square lattice is a local minimum for fixed $V$, then it is probably already the global minimizer for that $V$ in $\mathcal{Q}(V)$.     
\begin{figure}[!ht]
  \subfigure[For $V > V_1$ (dashed lines), the rock-salt structure is a local
  maximizer among 2-dimensional orthorhombic lattices of the given volume. For
  $V < V_1$ it is the unique minimizer.]{
    \includegraphics[width=.45\textwidth]{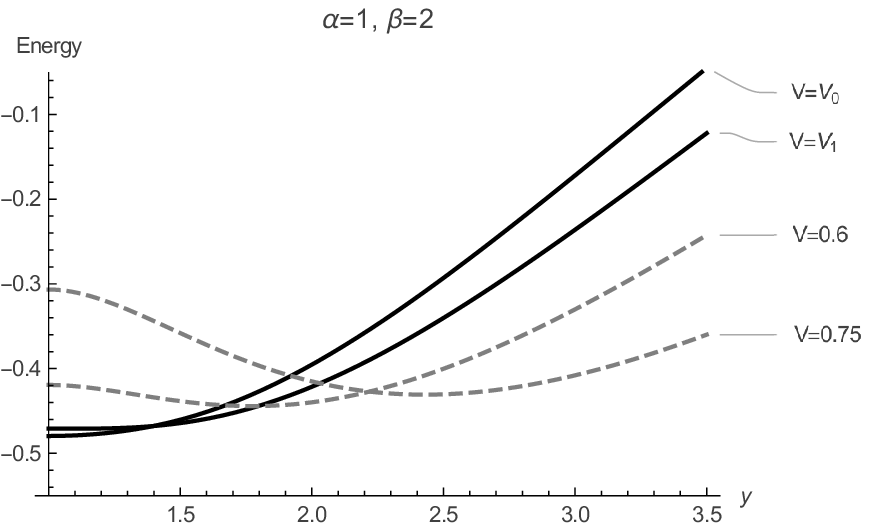} } \hfill
  \subfigure[{The second derivative of the function
    $y \mapsto E_{f_1,f_2}^\pm [V^{\frac{1}{2}} \left( y^{-1} \Z \times y \Z
    \right)]$ at $y=1$ as a function of $V$. The sign changes at $V = V_1$.}]{
    \includegraphics[width=.45\textwidth]{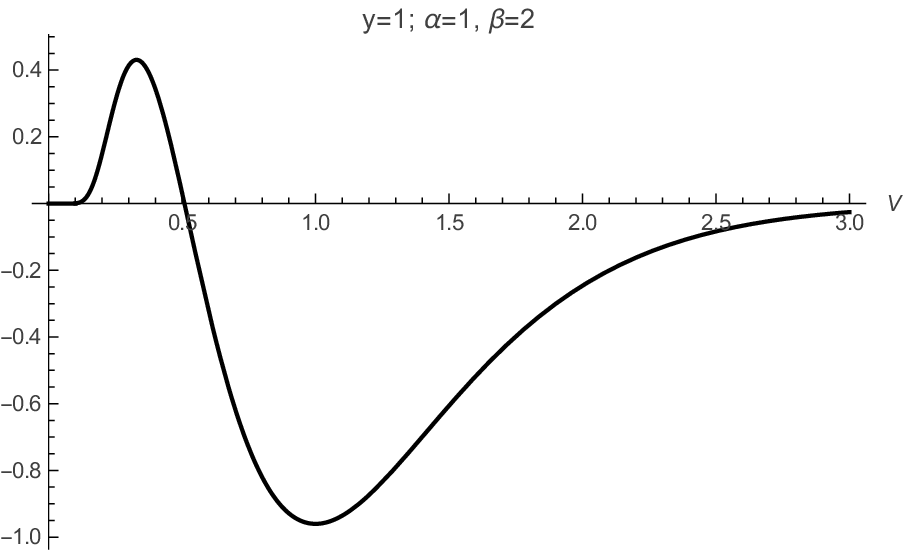} }
	\caption{\footnotesize{The energy $E_{f_1,f_2}^\pm [L]$ for rectangular lattices $V^{1/2} \left(y^{-1} \Z \times y \Z \right)$ of different volumes $V$ with Gaussian potentials as a function of $y$ (a). The value $V = V_0 \approx 0.464916$ gives the volume for which the rock-salt structure has minimal energy among $V > 0$. The value $V = V_1 \approx 0.508996$ is the threshold value, where the rock-salt structure turns from a minimizer into a local maximizer among orthorhombic lattices (b).}}
	\label{fig_Gauss_V}
\end{figure}

\subsubsection*{\textbf{b) Minimization of $E_{f_1,f_2}^\pm$ in $\mathcal{L}$}}

In the inverse power law case, by Thm.~\ref{thm-minenergyzetahomog}, it is possible to compare the energies of given  lattice ``shapes" (triangular, square, rhombic, rectangular) once $p,q$ are fixed. For instance, in Fig.~\ref{table-2d} we compare the minimal energy $\mathcal{E}_L^\pm$ given by \eqref{eq-minenergyL} of square and triangular lattices for different parameters $(p,q)$. This leads to good competitors for the minimization problem for $L\mapsto E_{f_1,f_2}^\pm[L]$ among all lattices.
\begin{figure}[!ht]
	\centering
	\begin{tabular}{|c|l|l|}
		\hline
		$(p,q)$		&\qquad\quad$\Z^2$		&\qquad\qquad$\mathsf{A}_2$\\
		\hline
		$(4,3)$		& $\mathcal{E}_{\Z^2}^\pm = -0.200328$,		& $\mathcal{E}_{\mathsf{A}_2}^\pm = -0.00730703$,\\
					& $V_{\Z^2} = 1.63374$				& $V_{\mathsf{A}_2} = 3.53775$\\
		\hline
		$(6,3)$		& $\mathcal{E}_{\Z^2}^\pm = -0.751062$,		& $\mathcal{E}_{\mathsf{A}_2}^\pm  = -0.165143$,\\
					& $V_{\Z^2} = 1.32499$				& $V_{\mathsf{A}_2} = 1.57651$\\
		\hline
	\end{tabular}
	\medskip
	\caption{\footnotesize{Comparision of the minimal energy of the square and triangular lattice with alternating charges and for power law potentials with fixed exponents $(p,q)$. For each lattice the optimal volume $V_L$ is chosen and the minimal energy $\mathcal{E}_L^\pm$ computed (cf. \eqref{def:VL}, \eqref{eq-minenergyL}). In both cases, the energy of the square lattice is smaller than the energy of the triangular lattice.}}\label{table-2d}
\end{figure}

In Fig.~\ref{fig_transition}, we have plotted the energy $L \mapsto E_{f_1,f_2}^\pm[V^{\frac{1}{2}}L]$ in the fundamental domain $\mathcal{D}$ defined by \eqref{def:Dlattices} which describes arbitrary two--dimensional lattices. The plots are for different values of $V$ and for the inverse power law case $(p,q)=(4,3)$. As $V$ increases, we observe a phase transition of the minimizer's shape of the form: triangular - rhombic - square - rectangular.
\begin{figure}[!htbp]
 	\centering
 	\subfigure[For $V=0.01$ the triangular lattice structure yields the minimum.]{
    \includegraphics[width=.45\textwidth,height=.5\textwidth]{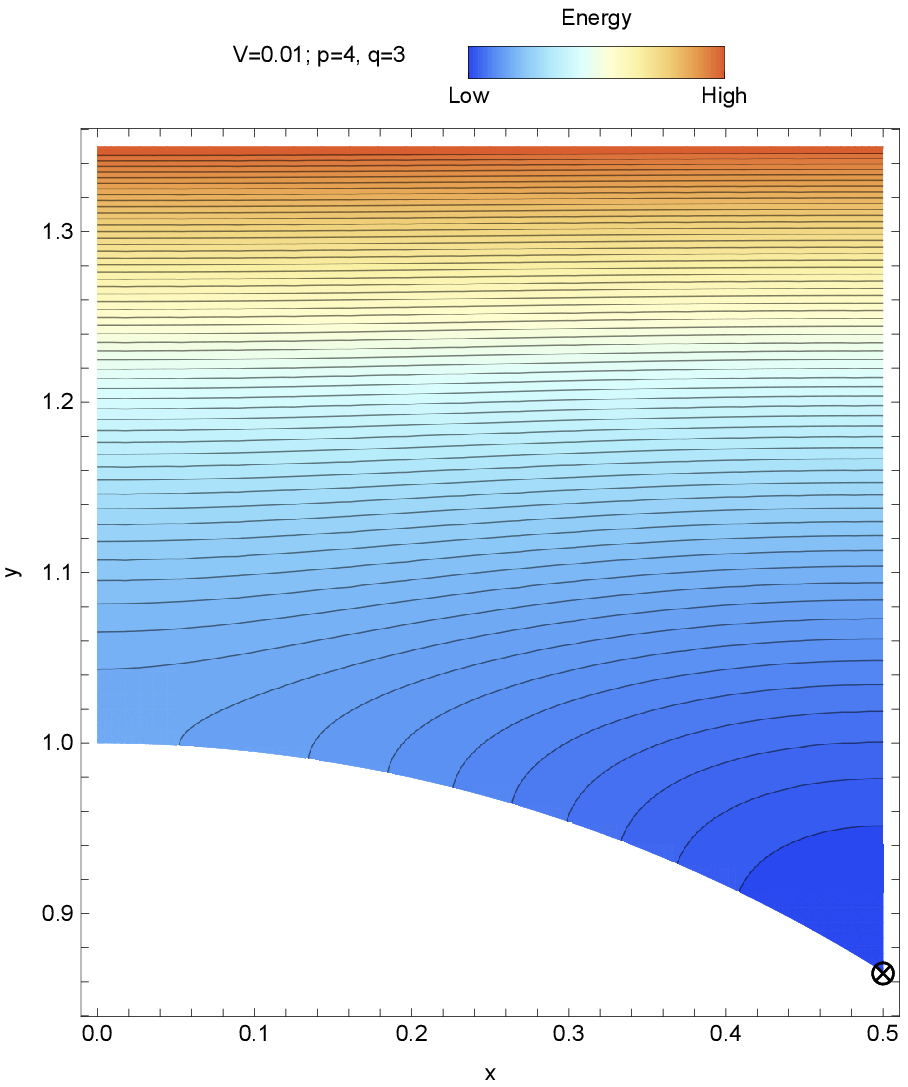}}
    \hfill
	\subfigure[For $V=0.3$ the minimizer is given by a rhombic lattice.]{
	\includegraphics[width=.45\textwidth,height=.5\textwidth]{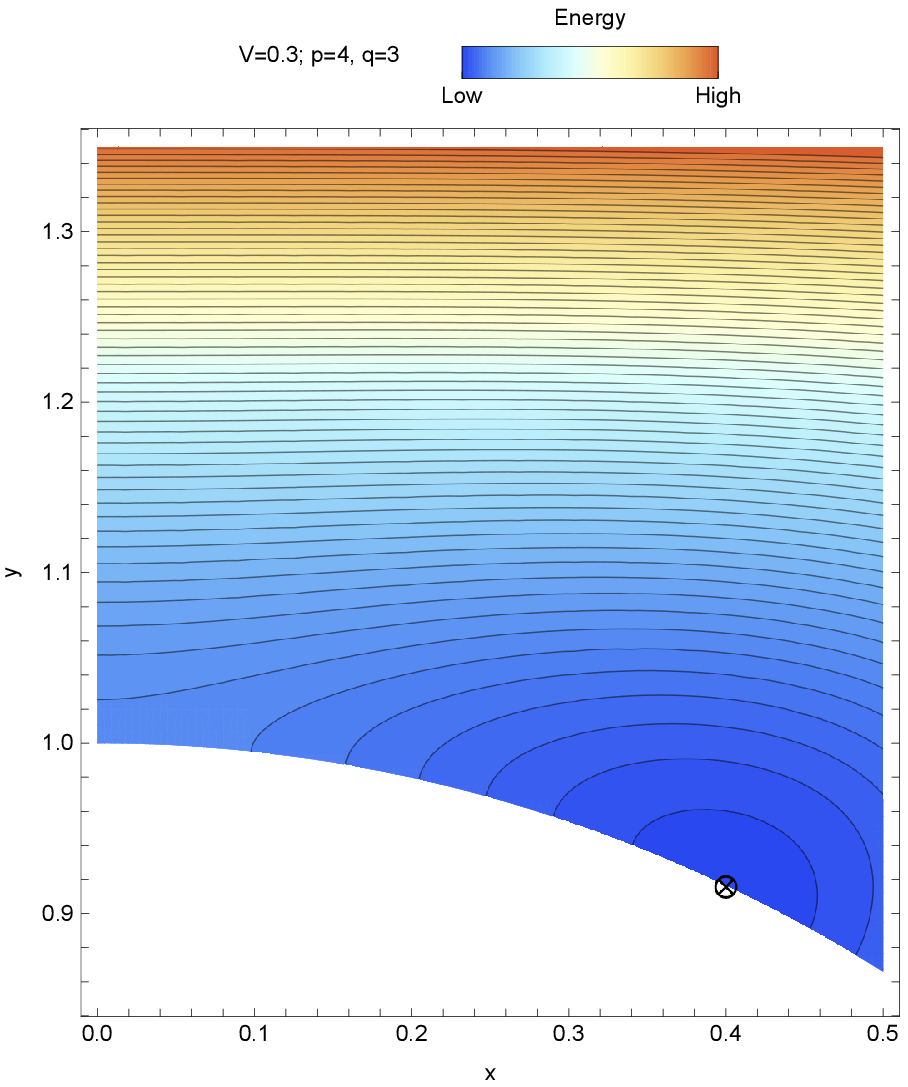}}
	\\
	\subfigure[For $V=1$ the minimizer is given by the square lattice.]{
    \includegraphics[width=.45\textwidth,height=.5\textwidth]{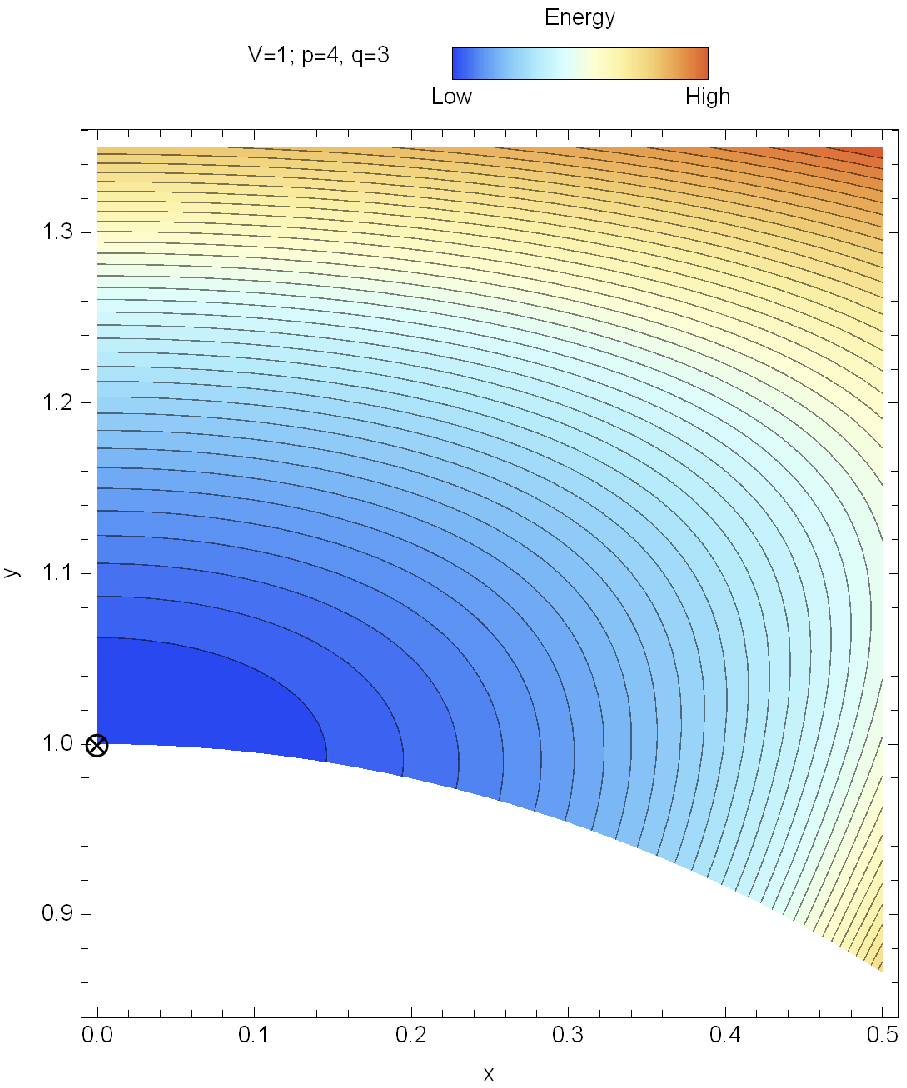}}
    \hfill
	\subfigure[For $V=1.85$ the minimizer is given by an orthorhombic lattice.]{
	\includegraphics[width=.45\textwidth,height=.5\textwidth]{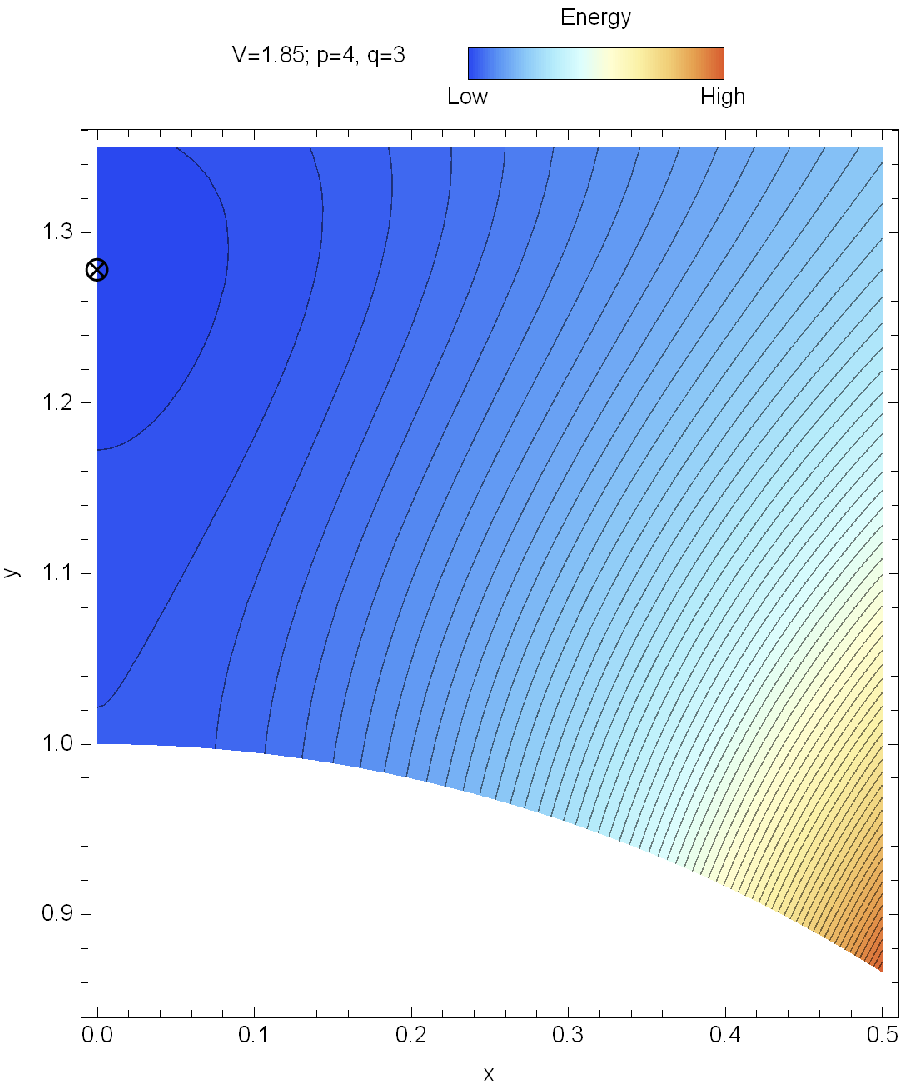}}
	\caption{\footnotesize{Contour plots of the energy $E_{f_1,f_2}^\pm$ in two dimensions for power law potentials $f_1$, $f_2$ with $(p,q) = (4,3)$ and for different volumes $V > 0$. The plot is restricted to the fundamental domain $\mathcal{D}$ defined in Rmk.~\ref{rmk:2dlattices}. The triangular lattice corresponds to $(x,y) = ( \tfrac{1}{2}, \tfrac{\sqrt{3}}{2})$, the square lattice to $(x,y) = (0,1)$. For the minimizer (marked by $\otimes$), we observe a phase transition of the type `triangular-rhombic-square-rectangular', as $V$ increases.}}
\label{fig_transition}
\end{figure}

Furthermore, in Fig.~\ref{fig_2d_rhombic}, we have plotted the energy for different rhombic lattices, including the square and the triangular lattice. The numerics suggest that the square lattice is optimal for this model.
\begin{figure}[!ht]
	\subfigure[Energy of rhombic lattices for the inverse power law case $(p,q)=(4,3)$.]{
    \includegraphics[width=0.425\textwidth]{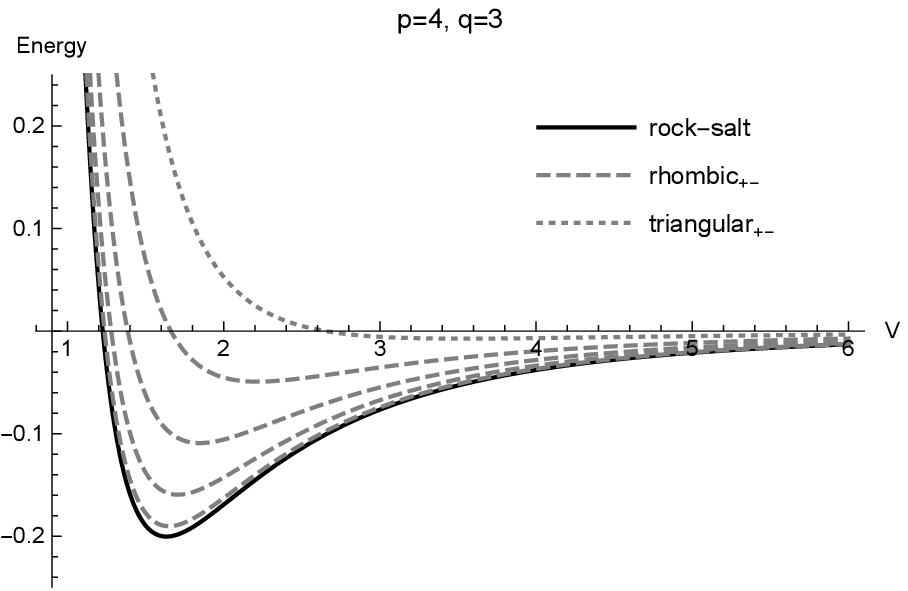} }
    \hfill
	\subfigure[Energy of rhombic lattices for the Gaussian case $(\alpha, \beta)= (1,2)$.]{
	\includegraphics[width=.425\textwidth]{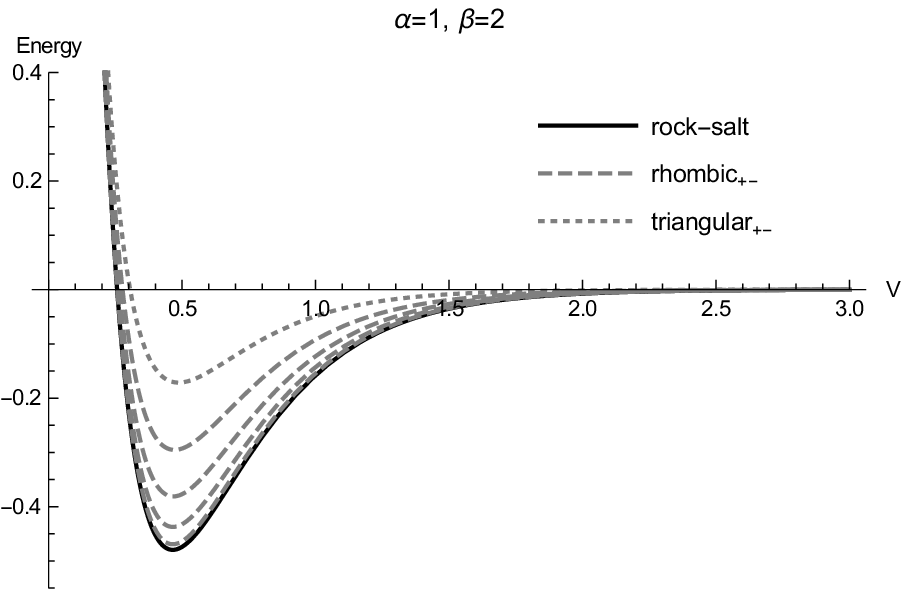} }
	\caption{\footnotesize{{Comparison of the energy, in the inverse power law and the Gaussian cases, amongst the rock-salt structure (solid line), the triangular structure with alternating charges (dotted line) and orthorhombic structures (dashed lines) with alternating charges.}}}\label{fig_2d_rhombic}
\end{figure}

This is also confirmed by Fig.~\ref{fig_2d_p4q3inD} where $L\mapsto \mathcal{E}_L^\pm$ is plotted, for $(p,q)=(4,3)$, on the fundamental domain $\mathcal{D}$ defined in Rmk.~\ref{rmk:2dlattices} and where the square lattice appears to be the unique minimum of this energy.
\begin{figure}[!ht]
	\includegraphics[width=.45\textwidth,height=.5\textwidth]{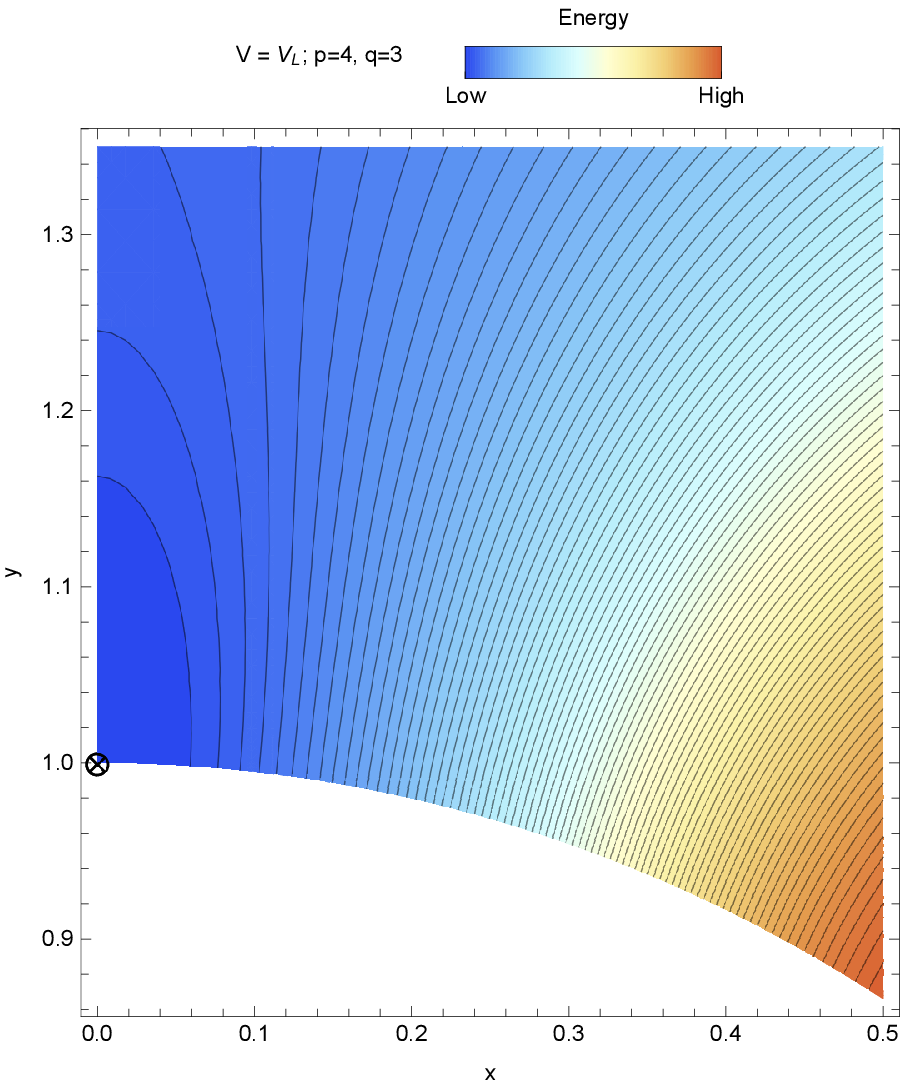} 
  \caption{\footnotesize{{Contour plot of $L\mapsto \mathcal{E}_L^\pm$ defined by \eqref{eq-minenergyL} in the inverse power law case with $(p,q)=(4,3)$. The minimum is achieved for the square lattice.}}}
	\label{fig_2d_p4q3inD}
\end{figure}
The same is observed for $(p,q)=(6,3)$. In the Gaussian case, we observe the same qualitative behavior. In particular, the fact that $\mathcal{E}_{\Z^2}^\pm<\mathcal{E}_{\mathsf{A}_2}^\pm$ appears to be true for any values of the parameters $p>q$, and the same is observed for the Gaussian case for any $\beta>\alpha$.

We note that the same phase transition for the minimizer with respect to the density has been observed for Lennard-Jones type potentials \cite{Beterloc,SamajTravenecLJ}, Morse type potentials \cite{LBMorse}, two-components Bose-Einstein condensates \cite{Mueller:2002aa} (see also Rmk.~\ref{rmk-BEC}) and 3-block co-polymers \cite{LuoChenWei}.

\subsubsection*{\textbf{c) Comparison of energies for lattices with optimal charge distribution}}
In Fig.~\ref{fig_Gauss_square_rhombic}, we compare the energy $V\mapsto E_{f_1,f_2}[V^{\frac{1}{2}}L,\varphi]$ for $L$ being the square and the triangular lattices and $\varphi$ being the alternate distribution of charges $\varphi_\pm$ or the optimal ``honeycomb-like" distribution of charges  $\varphi_{\tinyvarhexagon}$ for the triangular lattice found in \cite{BeterminKnuepfer-preprint}. In the inverse power law case $(p,q) = (4,3)$ and the Gaussian case $(\alpha, \beta) = (1,2)$, the rock-salt structure yields again the smallest value of the energy.
\begin{figure}[!htb]
	\subfigure[The inverse power law case for $(p,q) = (4,3)$.]{
		\includegraphics[width=.45\textwidth]{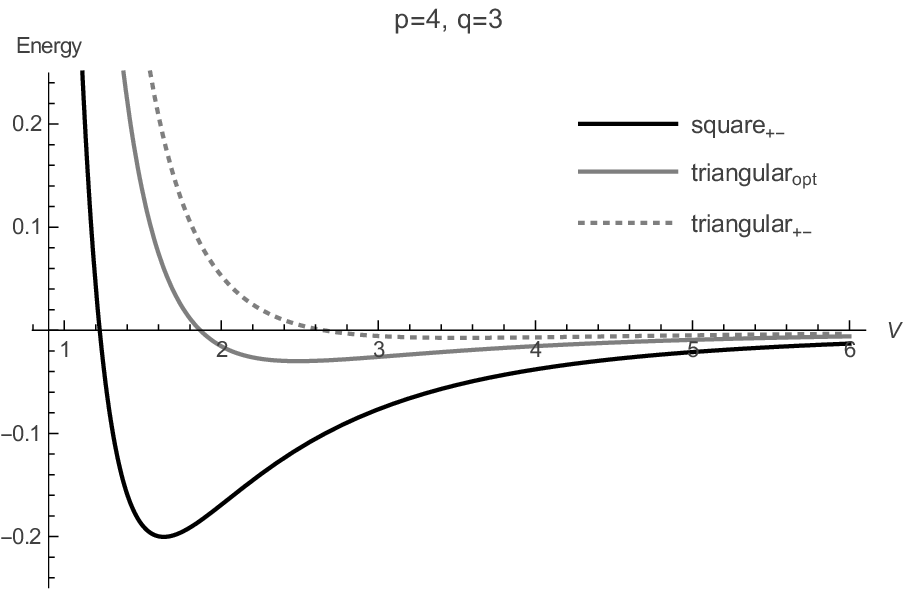}}
	\hfill
	\subfigure[The Gaussian case for $(\alpha, \beta) = (1,2)$.]{
    	\includegraphics[width=.45\textwidth]{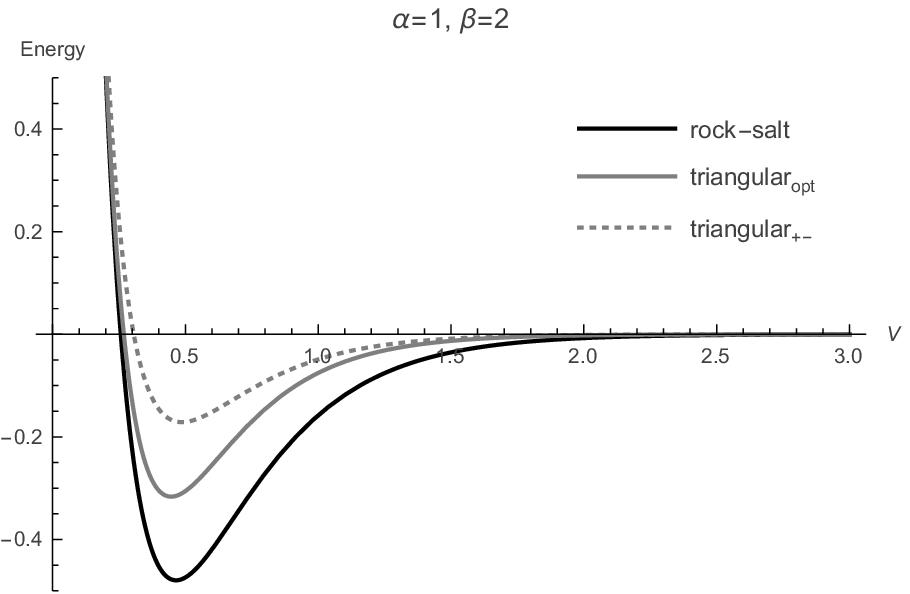}}
	\caption{\footnotesize{Comparison of $(\Z^2,\varphi_\pm)$ with $(\mathsf{A}_2, \varphi_{\tinyvarhexagon})$ (see Def.~\ref{def-chargedlattice}) and $(\mathsf{A}_2, \varphi_\pm)$. For fixed $V > 1$, the square lattice structure yields the smallest energy and the triangular structure with alternating charges gives the maximal energy among the compared structures.}}\label{fig_Gauss_square_rhombic}
\end{figure}

\subsection{Comparing the rock-salt structure to  BCC, FCC and $\mathsf{E}_8$}
In dimensions $d > 2$, the geometry of the fundamental domain $\mathcal{D}_d$ is much more complicated to describe (see, e.g., \cite{SarStromb}). Therefore, we only compare the energy of orthorhombic lattices and the special lattice structures BCC and FCC. As already explained in \cite{LBMorse}, these are the only possible lattices being critical points of $E_{f,0}$ in $\mathcal{L}(V)$. Thus, they are the main candidates for solving our minimization problem.

In Fig.~\ref{table-3d}, using Thm.~\ref{thm-minenergyzetahomog}, we give some values of the minimal energy $\mathcal{E}_L^\pm$ for the $\Z^3$, BCC and FCC structures with alternating charges in the case of power law potentials.
\begin{figure}[htb]
  \centering
    \begin{tabular}{|c|l|l|l|}
      \hline
      $(p,q)$		&\qquad\quad$\Z^3$				&\qquad\quad BCC					&\qquad\quad FCC\\
      \hline
      $(4,3)$		& $\mathcal{E}_{\Z^3}^\pm=-0.165476$,	& $\mathcal{E}_{\textnormal{BCC}}^\pm=-0.0000699095$,	& $\mathcal{E}_{\textnormal{FCC}}^\pm=-0.00685824$,\\
                        & $V_{\Z^3}=2.68968$		& $V_{\textnormal{BCC}}=44.3038$			& $V_{\textnormal{FCC}} = 7.89246$\\
      \hline
      $(6,3)$ 		& $\mathcal{E}_{\Z^3}^\pm=-0.924244$,	& $\mathcal{E}_{\textnormal{BCC}}^\pm=-0.0235324$, 	& $\mathcal{E}_{\textnormal{FCC}}^\pm=-0.240695$,\\
                        & $V_{\Z^3}=1.60949$		& $V_{\textnormal{BCC}}=3.415$ 				& $V_{\textnormal{FCC}} = 1.88408$\\
      \hline
    \end{tabular}
    \medskip
    \caption{\footnotesize{Values of the minimal energies $\mathcal{E}_L^\pm$
        and the corresponding volume $V_L$ for the cubic, BCC and FCC lattice
        structures in dimensions $d=3$ (cf.\ \eqref{eq-minenergyL},
        \eqref{def:VL}).} }
        \label{table-3d}
\end{figure}
We observe again that the cubic lattice $V_{\Z^3}^{1/3}\Z^3$ seems to be a good candidate for minimizing $E_{f_1,f_2}^\pm$ in $\mathcal{L}$. Furthermore, the fact that $\mathcal{E}_{\Z^3}^\pm<\mathcal{E}_{\textnormal{FCC}}^\pm<\mathcal{E}_{\textnormal{BCC}}^\pm$ appears to be true for any values of the parameters $p>q$ and the same holds in the Gaussian case for any $\beta>\alpha$.

In Fig.~\ref{fig_3d_special} we compare the energy of the 3-dimensional rock-salt structure with the FCC lattice structure with alternating charges as well as with the BCC structure with alternating charges and its (expected) optimal charge configuration (i.e., two cubic lattices of opposite charges shifted by the center of the primitive cube). The latter structure is found in Cesium Chloride ``CsCl".
\begin{figure}[!htb]
	{
		\includegraphics[width=.45\textwidth]{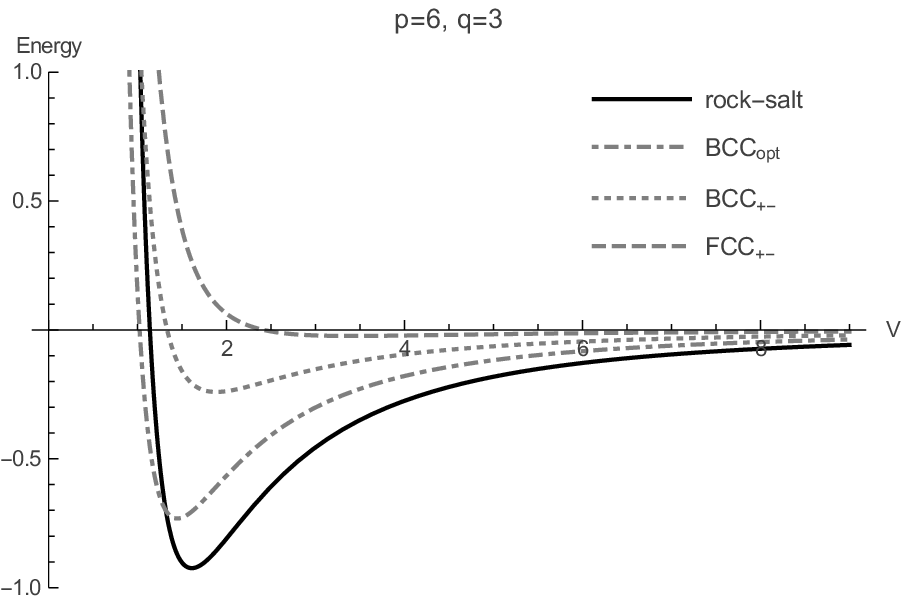}
	}
	\caption{\footnotesize{Comparison of energies, as functions of the volume $V$, in the inverse power laws case in 3 dimensions for $p=6$ and $q=3$. For the BCC lattice we also compared the (expected) optimal and alternating charge configurations ($\varphi_{opt}$ and $\varphi_\pm$, respectively). Among all compared lattices and all $V$, a cubic lattice yields the global minimum.}}
	\label{fig_3d_special}
\end{figure}

In the Gaussian case, we again compared the rock-salt structure to orthorhombic structures with alternate charge distribution. We numerically computed the value $V_1 \approx 0.623556$, which is (close to) the threshold value at which the rock-salt structure stops to be the minimizer. This was done by computing the eigenvalues of the Hessian matrix with respect to the lattice parameters. It was then evaluated at the parameters which yield the cubic lattice $V^{1/3}\Z^3$. In Fig.~\ref{fig_Gauss3d_hess} we show the behavior of the eigenvalues of the Hessian matrix as a function of $V$.
\begin{figure}[!htb]
	\centering
	\includegraphics[width=.45\textwidth]{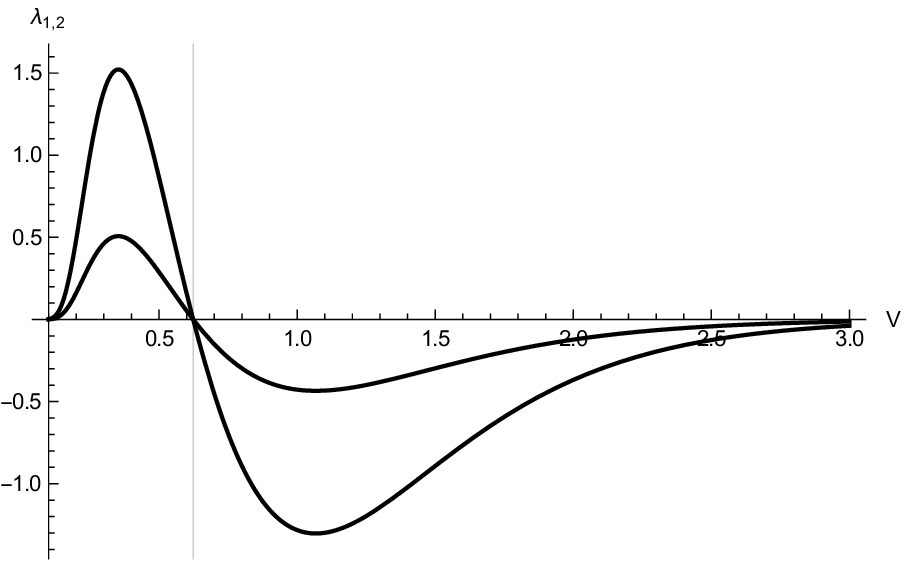}
	\caption{\footnotesize{Three-dimensional Gaussian case. The plot shows the eigenvalues $\lambda_{1,2}$ of the Hessian matrix of the energy $E_{f_1,f_2}^\pm$ evaluated at $a_1=a_2=1$ (rock-salt structure) in dependence of $V$. Above the threshold value $V_1 \approx 0.623556$, the rock-salt structure yields a global maximum for the energy model.}}
	\label{fig_Gauss3d_hess}
\end{figure}
We see that, among orthorhombic lattices, the cubic lattice seems to be a local minimizer (resp. maximizer) for $V<V_1$ (resp. $V>V_1$). In Fig.~\ref{fig_Gauss3d}, we illustrate a case where the rock-salt structure is a minimizer of the energy ($V < V_1$) as well as a case where the rock-salt structure is a local maximizer of the energy ($V > V_1$).

\begin{figure}[!htb]
	\subfigure[The rock-salt structure as a global minimizer.]{
	\includegraphics[width=.45\textwidth]{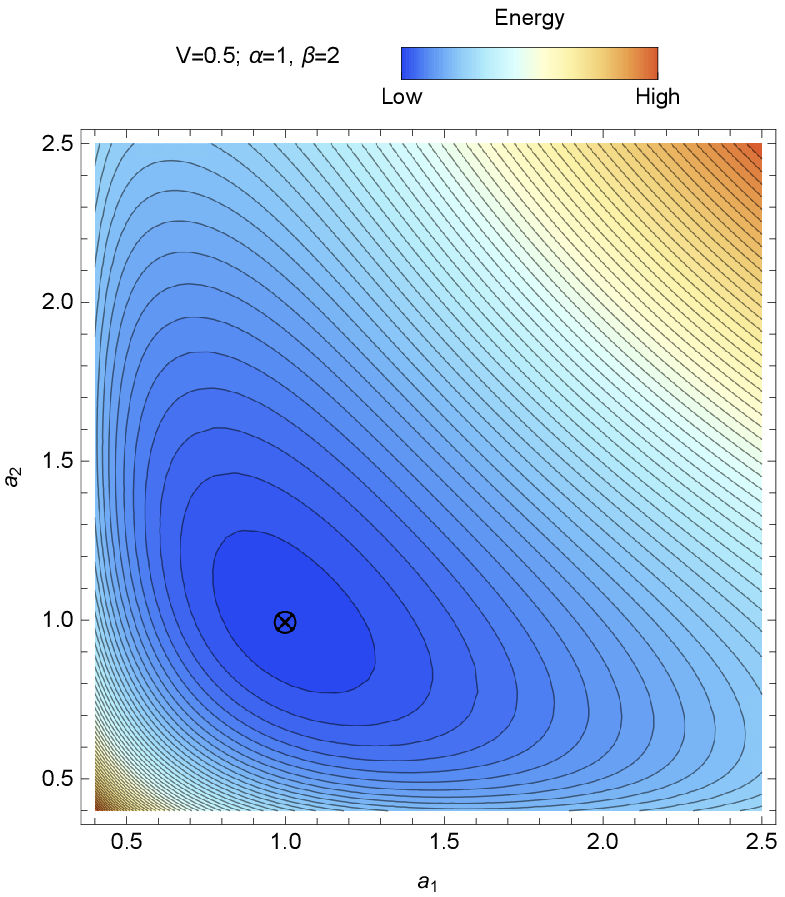}}
	\hfill
	\subfigure[The rock-salt structure as a local maximizer.]{
	\includegraphics[width=.45\textwidth]{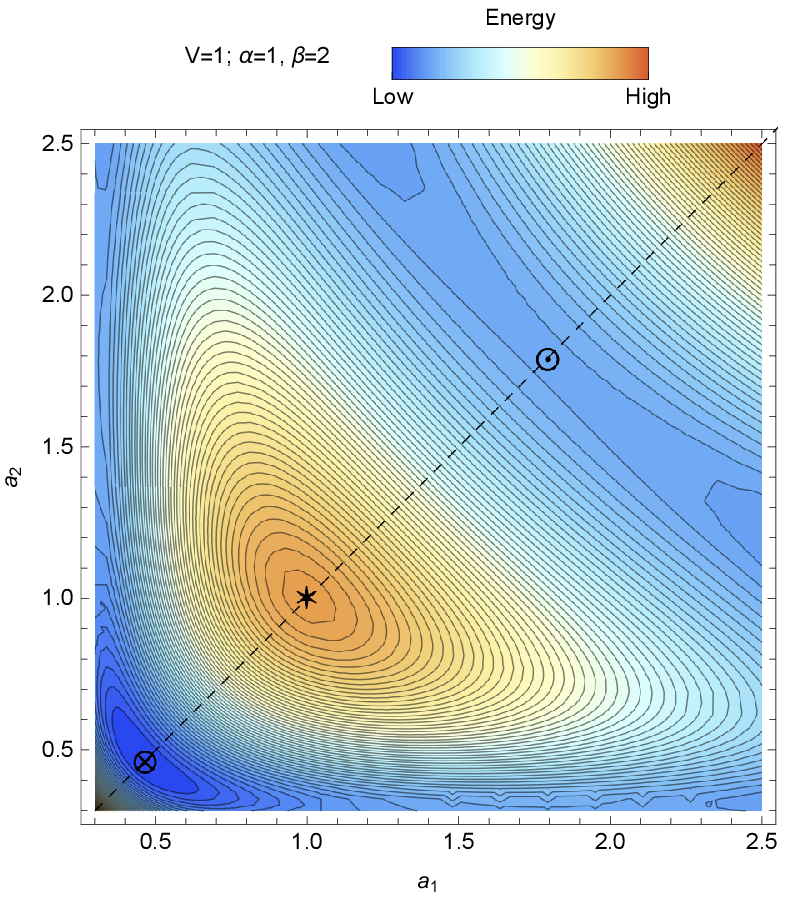}}
    \caption{\footnotesize{Plot of the energy $E_{f_1,f_2}^\pm$ for Gaussian interaction potentials with $\alpha = 1$, $\beta = 2$ for orthorhombic lattices of the form $V^{1/3} (a_1 \Z \times a_2 \Z \times (a_1 a_2)^{-1} \Z)$. (a) For volume V=0.5, the rock-salt structure $(a_1,a_2) = (1,1)$ is the global minimizer $\otimes$. (b) For $V = 1$, the rock-salt structure is a local maximum. The global minimum is marked with $\otimes$ and lies on the line $a_1 = a_2$. On the same line, there is a saddle point, marked with $\odot$. The rock-salt structure yields a local maximum $\ast$.}}\label{fig_Gauss3d}
\end{figure}

Let us now give our numerical findings in the eight-dimensional case. In the inverse power law case, using again Thm.~\ref{thm-minenergyzetahomog}, we have computed minimal energies $\mathcal{E}_L^\pm$ for rock-salt structures and the $\mathsf{E}_8$ lattice in dimensions $d = 8$, which is again a ``density stable" lattice in the sense of \cite{LBMorse}. Obviously, we have chosen this dimension and lattice structure due to the highly important universal minimality results proven in \cite{Viazovska,CKMRV,CKMRV2Theta}. A numerical comparison of the rock-salt structure and the Leech lattice with alternating charge distribution in dimension $d = 24$ was computationally not feasible, not even on GPU clusters, in a reasonable time with our methods. However, we assume that the 24-dimensional rock-salt structure will outperform the Leech lattice as well.\footnote{At this point, we would like to thank Pavol Harár from the Data Science Group at the Faculty of Mathematics of the University of Vienna for confirming our numerical results in dimension 8 and trying to speed up our numerics in dimension 24. Even though a significant speed-up was achieved, this was unfortunately still not enough to get results in a reasonable time.}

Nonetheless, we observe again that the cubic lattice $V_{\Z^8}^{1/8}\Z^8$ seems to be a good candidate for minimizing $E_{f_1,f_2}^\pm$ in $\mathcal{L}$, and, generally, to be a good candidate for the \textit{global} minimizer of our problem. In particular, the ``usual" minimizer, i.e., $\mathsf{E}_8$ does not seem to be the optimal candidates for our model in dimension 8. Also, we have no reason to believe that the Leech lattice $\Lambda_{24}$ will outperform the 24-dimensional rock-salt structure for our model.
\begin{figure}[!ht]
  \centering
  \begin{tabular}{|c|l|l|}
    \hline
    $(p,q)$ &\qquad\quad$\Z^8$ &\qquad\quad$\mathsf{E}_8$ \\
    \hline
    $(12,8)$ & $\mathcal{E}_{\Z^8}^\pm = -2.19656$ & $\mathcal{E}_{\mathsf{E}_8}^\pm = -0.00771459$  \\
            & $V_{\Z^8}=1.53947$ & $V_{\mathsf{E}_8}=1.52106$  \\
    \hline
  \end{tabular}
  \medskip
  \caption{\footnotesize{Values of the minimal energies
      $\mathcal{E}_L^\pm$ given by \eqref{eq-minenergyL} and the
      corresponding volume \eqref{def:VL} in dimension $d=8$ for comparing
      the cubic lattice with the minimizers of the Epstein zeta functions
      as proved in \cite{CKMRV2Theta}.}}
  \label{table-8d24d}
\end{figure}

\subsection{\textbf{Non--optimality of the rock-salt structure. }}
We notice an interesting phenomenon if $p-q$ and $\beta-\alpha$ are small enough. In that case, the global minimizer of $E_{f_1,f_2}^\pm$ is no longer a rock-salt structure. Indeed, Fig.~\ref{fig_Gaussip_nonopt} shows that the rock-salt structure is not the \textit{global} minimizer among orthorhombic lattices for the inverse power law case with parameters $(p,q)=(4,3.75)$. The same holds in the Gaussian case for $(\alpha,\beta)=(1.8,2)$. This behavior has also been observed in dimension $3$ for the same values of the parameters.
\begin{figure}[!htb]
	\subfigure[Non-optimality of the rock-salt structure in the inverse power law case.]{\includegraphics[width=.45\textwidth]{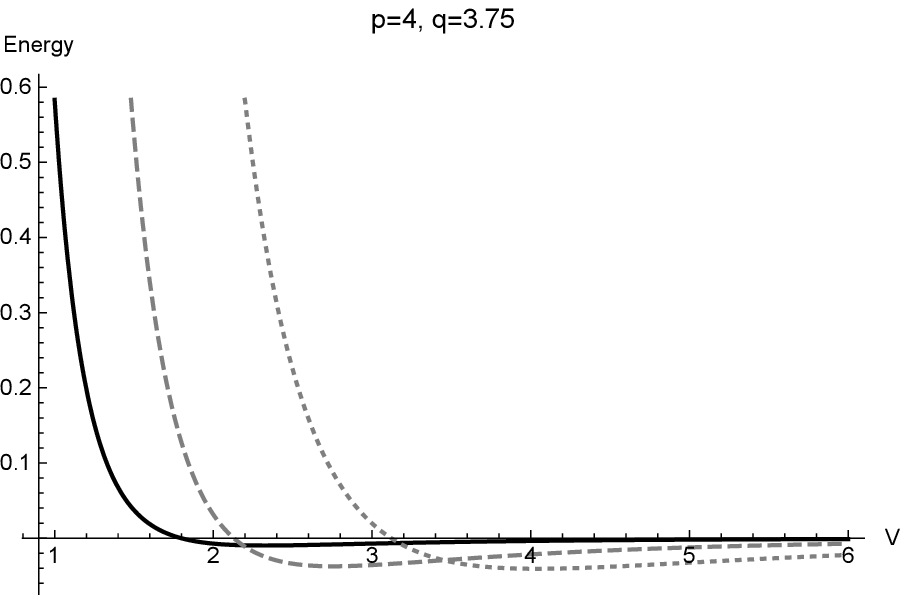}}
	\hfill
	\subfigure[Non-optimality of the rock-salt structure in the Gaussian case.]{\includegraphics[width=.45\textwidth]{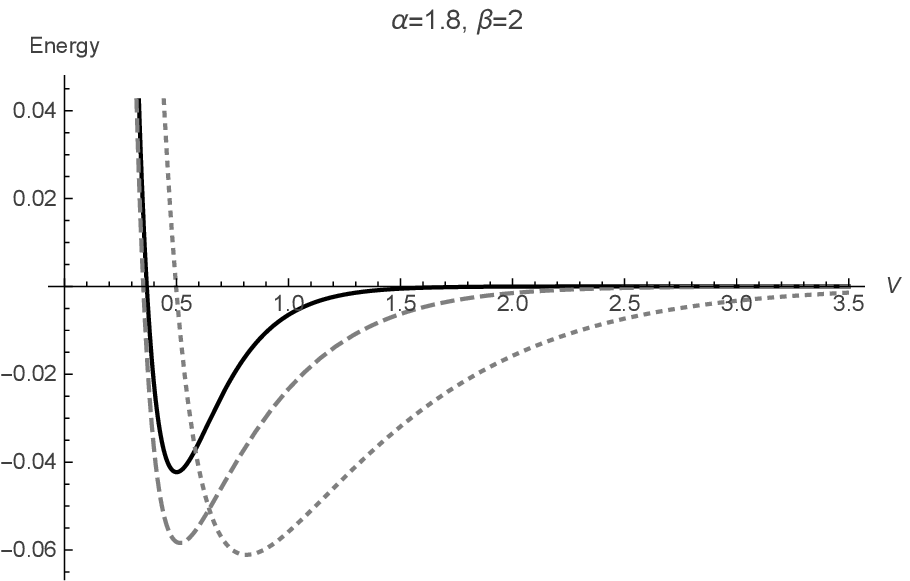}}
	\caption{\footnotesize{Comparison of $V\mapsto E_{f_1,f_2}^\pm[V^{\frac{1}{2}}L_a]$ for different rectangular lattices. The black line is the curve for the square lattice. For these values of the parameters $(p,q)=(4,3.75)$ and $(\alpha,\beta)=(1.8,2)$, the rock-salt structure is not the global minimizer of $E_{f_1,f_2}^\pm$ in $\mathcal{Q}$.}}\label{fig_Gaussip_nonopt}
\end{figure}

Furthermore, in the inverse power law case, it seems that $E^\pm_{f_1,f_2}$ does not have any minimizer when $p-q$ is small enough. In Fig.~\ref{fig_ipnonopt} we have compared the lowest possible energy of rectangular lattices by plotting $y\mapsto \mathcal{E}_{L_a}^\pm$, defined by \eqref{eq-minenergyL} for $p=4$ and $q\in \{3.15,3.75\}$. It appears that $\mathcal{E}_L^\pm$, and then $E_{f_1,f_2}^\pm$, do not have any minimizer in $\mathcal{Q}$ when $p-q$ is small enough, and the same is actually observed in $\mathcal{L}$.
\begin{figure}[!htb]
	\subfigure[The rock-salt structure as a local minimizer.]{
	\includegraphics[width=.45\textwidth]{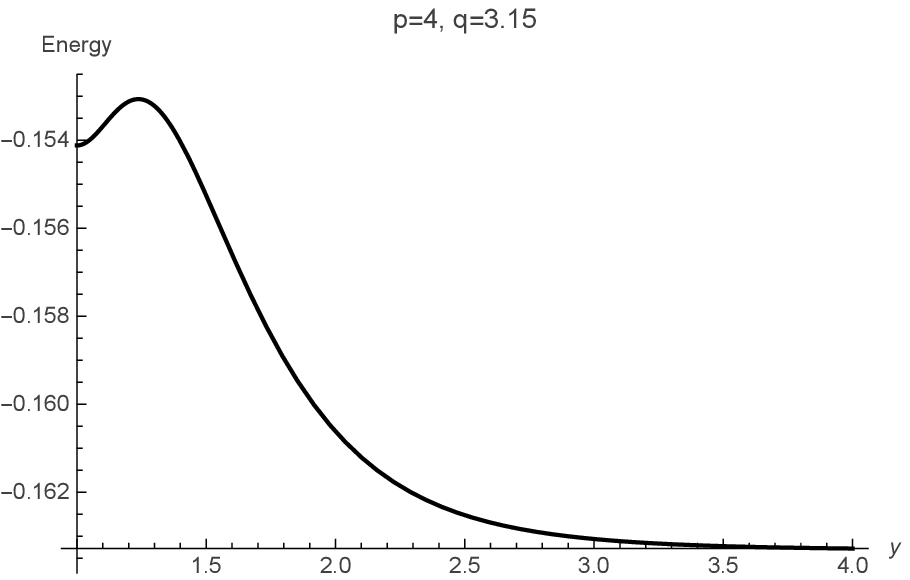}}
	\hfill
	\subfigure[The rock-salt structure as a global maximizer.]{
	\includegraphics[width=.45\textwidth]{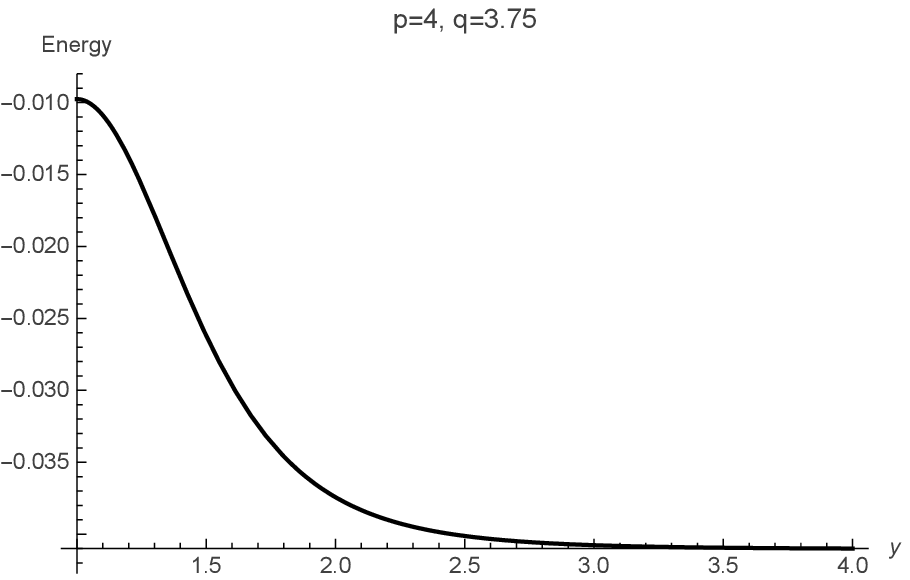}}
	\caption{\footnotesize{Plot of $y\mapsto \mathcal{E}_{L_a}^\pm$ when $L_a:=\Z(y,0)\oplus \Z(0,y^{-1})$, $p=4$ and $q\in \{3.15,3.75\}$. The energy $L\mapsto \mathcal{E}_L^\pm$ does not have a minimizer in $\mathcal{Q}$ for these values of the parameters.}}\label{fig_ipnonopt}
\end{figure}

\newpage
\bibliographystyle{plain}
\bibliography{mincharges}

\begin{thebibliography}{10}

\bibitem{AftBN}
A.~Aftalion, X.~Blanc, and F.~Nier.
\newblock {Lowest Landau level functional and Bargmann spaces for
  Bose--Einstein condensates}.
\newblock {\em J. Funct. Anal.}, 241:661--702, 2006.

\bibitem{Baernstein-1997}
A.~{Baernstein II}.
\newblock {A minimum problem for heat kernels of flat tori}.
\newblock {\em Contemp. Math.}, 201:227--243, 1997.

\bibitem{BetTheta15}
L.~B{\'e}termin.
\newblock {Two-dimensional Theta Functions and Crystallization among Bravais
  Lattices}.
\newblock {\em SIAM J. Math. Anal.}, 48(5):3236--3269, 2016.

\bibitem{Beterloc}
L.~B{\'e}termin.
\newblock {Local variational study of 2d lattice energies and application to
  Lennard-Jones type interactions}.
\newblock {\em Nonlinearity}, 31(9):3973--4005, 2018.

\bibitem{Beterminlocal3d}
L.~B{\'e}termin.
\newblock {Local optimality of cubic lattices for interaction energies}.
\newblock {\em Anal. Math. Phys.}, 9(1):403--426, 2019.

\bibitem{LBMorse}
L.~B{\'e}termin.
\newblock {Minimizing lattice structures for Morse potential energy in two and
  three dimensions}.
\newblock {\em J. Math. Phys.}, 60(10):102901, 2019.

\bibitem{BeterminKnuepfer-preprint}
L.~B{\'e}termin and H.~Kn{\"u}pfer.
\newblock {On {B}orn's conjecture about optimal distribution of charges for an
  infinite ionic crystal}.
\newblock {\em J. Nonlinear Sci.}, 28(5):1629--1656, 2018.

\bibitem{Crystbinary1d}
L.~B{\'e}termin, H.~Kn{\"u}pfer, and F.~Nolte.
\newblock Note on crystallization for alternating particle chains.
\newblock {\em J. Stat. Phys.}, 181(3):803--815, 2020.
\newblock \textit{Preprint. arXiv:1804.05743}.

\bibitem{BDLPSquare}
L.~B{\'e}termin, L.~De Luca, and M.~Petrache.
\newblock {Crystallization to the square lattice for a two-body potential}.
\newblock \textit{Preprint. arXiv:1907:06105}, 2019.

\bibitem{BeterminPetrache}
L.~B{\'e}termin and M.~Petrache.
\newblock {Dimension reduction techniques for the minimization of theta
  functions on lattices}.
\newblock {\em J. Math. Phys.}, 58:071902, 2017.

\bibitem{OptinonCM}
L.~B{\'e}termin and M.~Petrache.
\newblock {Optimal and non-optimal lattices for non-completely monotone
  interaction potentials.}
\newblock {\em Anal. Math. Phys.}, 9(4):2033--2073, 2019.

\bibitem{BlancLewin-2015}
X.~Blanc and M.~Lewin.
\newblock {The Crystallization Conjecture: A Review}.
\newblock {\em EMS Surv. in Math. Sci.}, 2:255--306, 2015.

\bibitem{Born-1921}
M.~Born.
\newblock {{\"U}ber elektrostatische Gitterpotentiale}.
\newblock {\em Z. Phys.}, 7:124--140, 1921.

\bibitem{CohnElkies}
H.~Cohn and N.~Elkies.
\newblock {New upper bounds on sphere packings I}.
\newblock {\em Ann. of Math.}, 157:689--714, 2003.

\bibitem{CohnKumar}
H.~Cohn and A.~Kumar.
\newblock {Universally optimal distribution of points on spheres}.
\newblock {\em J. Amer. Math. Soc.}, 20(1):99--148, 2007.

\bibitem{CKMRV}
H.~Cohn, A.~Kumar, S.~D. Miller, D.~Radchenko, and M.~Viazovska.
\newblock {The sphere packing problem in dimension 24}.
\newblock {\em Ann. of Math.}, 185(3):1017--1033, 2017.

\bibitem{CKMRV2Theta}
H.~Cohn, A.~Kumar, S.~D. Miller, D.~Radchenko, and M.~Viazovska.
\newblock {Universal optimality of the $E_8$ and Leech lattices and
  interpolation formulas}.
\newblock \textit{Preprint. arXiv:1902:05438}, 2019.

\bibitem{CoulLazzarini}
R.~Coulangeon and G.~Lazzarini.
\newblock {Spherical Designs and Heights of Euclidean Lattices}.
\newblock {\em J. Number Theory}, 141:288--315, 2014.

\bibitem{Coulangeon:2010uq}
R.~Coulangeon and A.~Sch{\"u}rmann.
\newblock {Energy Minimization, Periodic Sets and Spherical Designs}.
\newblock {\em Int. Math. Res. Not. IMRN}, pages 829--848, 2012.

\bibitem{CoulSchurm2018}
R.~Coulangeon and A.~Sch{\"u}rmann.
\newblock {Local energy optimality of periodic sets}.
\newblock \textit{Preprint. arXiv:1802.02072}, 2018.

\bibitem{DelucaFriesecke-2018}
L.~{De Luca} and G.~Friesecke.
\newblock {Crystallization in {T}wo {D}imensions and a {D}iscrete
  {G}auss--{B}onnet {T}heorem}.
\newblock {\em J. Nonlinear Sci.}, 28(1):69--90, 2018.

\bibitem{DeloneRysh}
B.~N. Delone and S.~S. Ryshkov.
\newblock {A Contribution to the Theory of the Extrema of a Multidimensional
  zeta-function}.
\newblock {\em Dokl. Akad. Nauk SSSR}, 173(4):991--994, 1967.

\bibitem{ELi}
W.~E and D.~Li.
\newblock {On the Crystallization of 2D Hexagonal Lattices}.
\newblock {\em Comm. Math. Phys.}, 286:1099--1140, 2009.

\bibitem{Ennola}
V.~Ennola.
\newblock {On a Problem about the Epstein Zeta-Function}.
\newblock {\em Math. Proc. Cambridge Philos. Soc.}, 60:855--875, 1964.

\bibitem{Ewald1}
P.~Ewald.
\newblock {Die Berechnung optischer und elektrostatischer Gitterpotentiale}.
\newblock {\em Ann. Phys.}, 64:253--287, 1921.

\bibitem{Faulhuber_Rama_2019}
M.~Faulhuber.
\newblock {An Application of Hypergeometric Functions to Heat Kernels on
  Rectangular and Hexagonal Tori and a ``Weltkonstante" - Or - How Ramanujan
  Split Temperatures}.
\newblock {\em Ramanujan J.}, 2020.

\bibitem{Faulhuber:2016aa}
M.~Faulhuber and S.~Steinerberger.
\newblock {Optimal Gabor frame bounds for separable lattices and estimates for
  Jacobi theta functions}.
\newblock {\em J. Math. Anal. Appl.}, 445(1):407--422, 2017.

\bibitem{FauSte19}
M.~Faulhuber and S.~Steinerberger.
\newblock {An Extremal Property of the Hexagonal Lattice}.
\newblock {\em J. Stat. Phys.}, 177(2):285--298, 2019.

\bibitem{TheilFlatley}
L.~Flatley and F.~Theil.
\newblock {Face-Centred Cubic Crystallization of Atomistic Configurations}.
\newblock {\em Arch. Ration. Mech. Anal.}, 219(1):363--416, 2015.

\bibitem{Friedrich:2018aa}
M.~Friedrich and L.~Kreutz.
\newblock {Crystallization in the hexagonal lattice for ionic dimers}.
\newblock {\em Math. Models Methods Appl. Sci.}, 29(10):1853--1900, 2019.

\bibitem{FrieKreutSquare}
M.~Friedrich and L.~Kreutz.
\newblock {Finite crystallization and Wulff shape emergence for ionic compounds
  in the square lattice}.
\newblock {\em Nonlinearity}, 33(3):1240--1296, 2020.

\bibitem{Gruber}
P.~M. Gruber.
\newblock {Application of an Idea of Voronoi to Lattice Zeta Functions}.
\newblock {\em Proc. Steklov Inst. Math.}, 276:103--124, 2012.

\bibitem{SaffLongRange}
D.~P. Hardin, E.~B. Saff, and B.~Simanek.
\newblock {Periodic Discrete Energy for Long-Range Potentials}.
\newblock {\em J. Math. Phys.}, 55(12):123509, 2014.

\bibitem{Rad2}
R.~C. Heitmann and C.~Radin.
\newblock {The Ground State for Sticky Disks}.
\newblock {\em J. Stat. Phys.}, 22:281--287, 1980.

\bibitem{ReviewvorticesBEC}
K.~Kasamatsu, M.~Tsubota, and M.~Ueda.
\newblock {Vortices in Multicomponent Bose-Einstein Condensates}.
\newblock {\em Int. J. Mod. Phys. B}, 19(11):1835--1904, 2005.

\bibitem{LimTeo}
S.~C. Lim and L.~P. Teo.
\newblock {On the minima and convexity of Epstein zeta function}.
\newblock {\em J. Math. Phys.}, 49(7):073513, 2008.

\bibitem{LuoChenWei}
S.~Luo, X.~Ren, and J.~Wei.
\newblock {Non-hexagonal lattices from a two species interacting system}.
\newblock {\em SIAM Journal on Mathematical Analysis}, 52(2):1903--1942, 2020.

\bibitem{Stef1}
E.~Mainini, P.~Piovano, and U.~Stefanelli.
\newblock {Finite crystallization in the square lattice}.
\newblock {\em Nonlinearity}, 27:717--737, 2014.

\bibitem{Stef2}
E.~Mainini and U.~Stefanelli.
\newblock {Crystallization in carbon nanostructures}.
\newblock {\em Comm. Math. Phys.}, 328:545--571, 2014.

\bibitem{Mont}
H.~L. Montgomery.
\newblock {Minimal Theta Functions}.
\newblock {\em Glasg. Math. J.}, 30(1):75--85, 1988.

\bibitem{Mueller:2002aa}
E.~J. Mueller and T.-L. Ho.
\newblock {Two-Component Bose-Einstein Condensates with a Large Number of
  Vortices}.
\newblock {\em Phys. Rev. Lett.}, 88(18), 2002.

\bibitem{NaumannBook}
R.~J. Naumann.
\newblock {\em {Introduction to the Physics and Chemistry of Materials}}.
\newblock CRC Press, 2008.

\bibitem{NonnenVoros}
S.~Nonnenmacher and A.~Voros.
\newblock {Chaotic Eigenfunctions in Phase Space}.
\newblock {\em J. Stat. Phys.}, 92:431--518, 1998.

\bibitem{Ewaldpolytropic}
O.~N. Osychenko, G.~E. Astrakharchik, and J.~Boronat.
\newblock {Ewald method for polytropic potentials in arbitrary dimensionality}.
\newblock {\em Mol. Phys.}, 110(4):227--247, 2012.

\bibitem{Rad3}
C.~Radin.
\newblock {The Ground State for Soft Disks}.
\newblock {\em J. Stat. Phys.}, 26(2):365--373, 1981.

\bibitem{SarStromb}
P.~Sarnak and A.~Str{\"o}mbergsson.
\newblock {Minima of Epstein's Zeta Function and Heights of Flat Tori}.
\newblock {\em Invent. Math.}, 165:115--151, 2006.

\bibitem{SteinWeiss}
E.~Stein and G.~Weiss.
\newblock {\em {Introduction to {F}ourier analysis on {E}uclidean spaces}}.
\newblock Princeton University Press, Princeton, N.J., 1971.
\newblock Princeton Mathematical Series, No. 32.

\bibitem{Crystal}
F.~Theil.
\newblock {A Proof of Crystallization in Two Dimensions}.
\newblock {\em Comm. Math. Phys.}, 262(1):209--236, 2006.

\bibitem{SamajTravenecLJ}
I.~Trav\v{e}nec and L.~\v{S}amaj.
\newblock {Two-dimensional Wigner crystals of classical Lennard-Jones
  particles}.
\newblock {\em J. Phys. A: Math. Theor.}, 52(20):205002, 2019.

\bibitem{Viazovska}
M.~Viazovska.
\newblock {The sphere packing problem in dimension 8}.
\newblock {\em Ann. of Math.}, 185(3):991--1015, 2017.

\bibitem{WhiWat69}
E.~T. Whittaker and G.~N. Watson.
\newblock {\em {A Course of Modern Analysis}}.
\newblock Cambridge University Press, reprinted edition, 1969.

\bibitem{Mathematica}
{Wolfram Research{,} Inc.}
\newblock {\em Mathematica, {V}ersion 12.0}.
\newblock Champaign, IL, 2019.

\end{thebibliography}
\end{document}